\documentclass{lmcs}
\pdfoutput=1

\usepackage{lastpage}
\lmcsdoi{21}{1}{22}
\lmcsheading{}{\pageref{LastPage}}{}{}%
{Jan.~24,~2024}{Mar.~06,~2025}{}


\usepackage[arrow,matrix,frame,curve,cmtip]{xy}\CompileMatrices

\usepackage{stmaryrd}
\usepackage{hyperref}

\usepackage{float}
\usepackage{physics}
\usepackage{amsthm}
\usepackage{amssymb}
\usepackage{enumitem}
\usepackage{mathtools}
\usepackage[utf8]{inputenc}
\usepackage{tikz}
\usetikzlibrary{matrix, arrows.meta, positioning, calc,arrows}
\usepackage[usestackEOL]{stackengine}

\DeclareFontFamily{U}{mathx}{\hyphenchar\font45}
\DeclareFontShape{U}{mathx}{m}{n}{
	<5> <6> <7> <8> <9> <10>
	<10.95> <12> <14.4> <17.28> <20.74> <24.88>
	mathx10
}{}
\DeclareSymbolFont{mathx}{U}{mathx}{m}{n}
\DeclareFontSubstitution{U}{mathx}{m}{n}
\DeclareMathAccent{\widecheck}{0}{mathx}{"71}
\DeclareMathAccent{\wideparen}{0}{mathx}{"75}

\DeclareFontFamily{OT1}{pzc}{}
\DeclareFontShape{OT1}{pzc}{m}{it}{<-> s * [1.200] pzcmi7t}{}
\DeclareMathAlphabet{\mathpzc}{OT1}{pzc}{m}{it}
	\newcommand\functorop[1][l]{\csname#1functor\endcsname}
	\newcommand\lfunctorop[3]{%
		\setbox0=\hbox{$#2$}%
		\kern\wd0%
		\ensurestackMath{\Centerstack[c]{#1\\ \mathllap{#2\;\,}\mathclap{\DownArrow}\\#3}}%
	}		
	\newcommand\rfunctorop[3]{%
		\setbox0=\hbox{$#2$}%
		\ensurestackMath{\Centerstack[c]{#1\\\mathclap{\UpArrow}\mathrlap{\,\;#2}\\#3}}%
		\kern\wd0%
	}
	
	\setstackgap{L}{1.3\normalbaselineskip}
	\newcommand\UpArrow{\rotatebox[origin=c]{90}{$\longrightarrow$\,}}
	\newcommand\DownArrow{\rotatebox[origin=c]{-90}{$\longrightarrow$\,}}
	\newcommand\functor[1][l]{\csname#1functor\endcsname}
	\newcommand\lfunctor[3]{%
		\setbox0=\hbox{$#2$}%
		\kern\wd0%
		\ensurestackMath{\Centerstack[c]{#1\\ \mathllap{#2\;\,}\mathclap{\DownArrow}\\#3}}%
	}
	\newcommand\rfunctor[3]{%
		\setbox0=\hbox{$#2$}%
		\ensurestackMath{\Centerstack[c]{#1\\\mathclap{\DownArrow}\mathrlap{\,\;#2}\\#3}}%
		\kern\wd0%
	}
	
	\setstackgap{L}{1.3\normalbaselineskip}

\newcommand{\pad}[1]{\mathcal{#1}_{\mathsf{pa}}}
\newcommand{\ad}[1]{\mathcal{#1}_{\mathsf{a}}}

	\newcommand{\gr}{\catname{Graph}}
	\newcommand{\dgr}{\catname{SGraph}}
	\newcommand{\dg}{\catname{DAG}}
	\newcommand{\rt}{\mathsf{dcl_s}}
	\newcommand{\rta}{\mathsf{dcl}}
	\newcommand{\rtd}{\mathsf{dcl_{d}}}
	\newcommand{\catname}[1]{\mathbf{#1}}

	\newcommand{\arr}[1]{\mathsf{Mor}(\catname{#1})}
	
	\newcommand{\pred}[1]{{\downarrow}#1}

	\newcommand{\id}[1]{\mathsf{id}_{#1}}
	\newcommand{\comma}[2]{#1\hspace{1pt} {\downarrow}\hspace{1pt} #2}
	\newcommand{\cma}[2]{\mathcal{#1}\hspace{1pt} {\downarrow}\hspace{1pt} \mathcal{#2}}

	\newcommand{\true}{\mathsf{t}}

	\newcommand{\mor}{\mathsf{Mor}}
	\newcommand{\mon}{\mathsf{Mono}}
	\newcommand{\reg}{\mathsf{Reg}}

	\def\A{\textbf {\textup{A}}}
	\def\B{\textbf {\textup{B}}}
	\def\C{\textbf {\textup{C}}}
	
	\def\X{\textbf {\textup{X}}}
	\def\Y{\textbf {\textup{Y}}}
	
	\newcommand{\Set}{\textbf{\textup{Set}}}
	
	\def\j{\mathfrak{j}}
	\def\jp{\mathfrak{j}_{\mathcal{M}, \mathcal{N}}}
	\newcommand{\sh}{\mathbf{Sh}}
	
\newcommand{\yo}{\mathcal{Y}}

	\newcommand{\sub}{\mathsf{Sub}}
	\newcommand{\msub}[1]{\mathsf{Sub}_{\mathcal{#1}}}

	\newtheorem*{lma1}{Lemma $\mathbf{2.5}$}
	
	\newtheorem*{lma2}{Lemma $\mathbf{2.7}$}
	
	\newcommand{\renewtheorem}[1]{%
		\expandafter\let\csname #1\endcsname\relax
		\expandafter\let\csname c@#1\endcsname\relax
		\expandafter\let\csname end#1\endcsname\relax
		\newtheorem{#1}%
	}
	\usepackage[capitalize]{cleveref}

	\title {On the axioms of $\mathcal{M}, \mathcal{N}$-adhesive categories}

\begin{document}
\author[D.~Castelnovo]{Davide Castelnovo\lmcsorcid{0000-0002-5926-5615}}[a]	
\author[M.~Miculan]{Marino Miculan\lmcsorcid{0000-0003-0755-3444}}[b]	

\address{Department of Mathematics, University of Padova, Padova, Italy.}	
\email{davide.castelnovo@math.unipd.it}

\address{Department of Mathematics, Computer Science and Physics, University of Udine, Udine, Italy.}	
\email{marino.miculan@uniud.it}

\begin{abstract} 
	Adhesive and quasiadhesive categories provide a general framework for the study of algebraic graph rewriting systems. 
	In a quasiadhesive category any two regular subobjects have a join which is again a regular subobject. Vice versa, if regular monos are \emph{adhesive}, then the existence of a regular join for any pair of regular subobjects entails quasiadhesivity. 
	It is also known that (quasi)adhesive categories can be embedded in a Grothendieck topos via a functor preserving pullbacks and pushouts along (regular) monos.
	
	In this paper we extend these results to \emph{$\mathcal{M}, \mathcal{N}$-adhesive categories}, a concept which generalizes the notion of (quasi)adhesivity.	
	We introduce the notion of \emph{$\mathcal{N}$-adhesive morphism}, which allows us to express $\mathcal{M}, \mathcal{N}$-adhesivity as a condition on the subobjects' posets. Moreover, $\mathcal{N}$-adhesive morphisms allows us to show how an $\mathcal{M},\mathcal{N}$-adhesive category can be embedded into a Grothendieck topos, preserving pullbacks and $\mathcal{M}, \mathcal{N}$-pushouts. 
\end{abstract}
\maketitle

\section{Introduction}
Since their introduction \cite{lack2005adhesive}, adhesive and quasiadhesive categories have provided a powerful and elegant  framework for studying and guaranteeing useful properties of algebraic graph rewriting systems \cite{behr2022fundamentals,corradini1997algebraic,ehrig2006fundamentals}. 
Besides being useful for the study of rewriting systems, adhesivity and quasiadhesivity entail other properties of the underlying categories \cite{johnstone2007quasitoposes,lack2005adhesive}. In a quasiadhesive category any two regular subobjects (i.e. subobjects represented by a regular mono) have a join which is again a regular subobject. 
Vice versa, if regular monos are \emph{adhesive}, then the existence of a regular join for any pair of regular subobjects entails quasiadhesivity \cite{garner2012axioms}. 
Moreover, any (quasi)adhesive category can be embedded in a Grothendieck topos via a functor preserving pullbacks and pushouts along (regular) monomorphisms \cite{garner2012axioms,lack2011embedding}, thus justifying the slogan that ``a category is (quasi)adhesive if pushouts of (regular) monomorphisms exist and interact with pullbacks as they do in a topos''.

In order to deal with situations with other conditions on the rewriting rules or on the allowed matchings, adhesivity and quasiadhesivity have been generalized with the introduction of $\mathcal{M}$- and \emph{$\mathcal{M}, \mathcal{N}$-adhesive categories} \cite{azzi2019essence,CastelnovoGM22,habel2012mathcal,peuser2016composition}, where $\mathcal{M}$ and $\mathcal{N}$ are various classes of monomorphisms.
While results concerning rewriting,  like the Church-Rosser Theorem and the Parallelism Theorem, can be proved also in this setting, a study of the consequences on the underlying categorical structure entailed by  $\mathcal{M}, \mathcal{N}$-adhesivity is lacking. 

To this end, in this paper we first establish a relationship between $\mathcal{M}, \mathcal{N}$-adhesivity and the existence of some binary suprema in the poset of subobjects of a given object.  Generalizing the approach of \cite{garner2012axioms}, we introduce the notion of \emph{$\mathcal{N}$-adhesive morphism}; then we show that, if $\mathcal{M}$ and $\mathcal{N}$ are nice enough,  $\mathcal{M}, \mathcal{N}$-adhesivity entails the existence of suprema for some pairs of subobjects; vice versa, the existence of these suprema, together with every arrow in $\mathcal{M}$ being $\mathcal{N}$-adhesive, is enough to guarantee $\mathcal{M}, \mathcal{N}$-adhesivity.

Moreover, the framework of $\mathcal{N}$-adhesive morphisms allows us to generalize also the embedding results provided in \cite{garner2012axioms,lack2011embedding}. As in the (quasi)adhesive cases, and under some hypotheses on the classes $\mathcal{M}$ and $\mathcal{N}$, an $\mathcal{M},\mathcal{N}$-adhesive category admits a full and faithful functor into a Grothendieck topos, which preserves pullbacks and $\mathcal{M}, \mathcal{N}$-pushouts.
Thus, this suggests the new slogan that ``a category is $\mathcal{M},\mathcal{N}$-adhesive if $\mathcal{M}, \mathcal{N}$-pushouts exist and interact with pullbacks as they do in a topos''.

\paragraph{Synopsis} 
In \cref{sec:defi} we will recall the notion of Van Kampen square and the definition of $\mathcal{M},\mathcal{N}$-adhesive category, showing how it relates with $\mathcal{M}$-adhesivity and (quasi)adhesivity. 
\cref{sec:n-pread-mor} introduces the notion of $\mathcal{N}$-(pre)adhesive morphism, which is used in \cref{sec:uniad} to provide a formula to compute certain suprema in the poset of subobjects of an $\mathcal{M}, \mathcal{N}$-category. \cref{sec:aduni} goes on the other direction: starting from the existence of certain suprema we can prove $\mathcal{M}, \mathcal{N}$-adhesivity. 
Using the results of these sections, in \cref{sec:top} we will  show that, under some hypotheses on $\mathcal{M}$ and $\mathcal{N}$, every $\mathcal{M}, \mathcal{N}$-adhesive category can be embedded in a Grothendieck topos via a functor preserving pullbacks and $\mathcal{M}, \mathcal{N}$-pushouts. 
Conclusions and suggestions for further work are in \cref{sec:concl}.

For the sake of brevity, proofs of straightforward results are omitted, but they can be found in the extended version of this paper \cite{castelnovo2024v4}.


\section{$\mathcal{M}, \mathcal{N}$-adhesive categories}\label{sec:defi}
In this section we recall the basic definitions and results about $\mathcal{M}, \mathcal{N}$-adhesive categories. We begin with some notational conventions which will be used throughout this paper. 

\begin{itemize}\item 
	Given a category $\X$ we will not distinguish notationally between $\X$ and its class of objects: so that ``$X\in \X$'' means that $X$ belongs to the class of objects of $\X$.  
	\item 
	If $1$ is a terminal object in a category $\X$,  the unique arrow $X\to 1$ from another object $X$ will be denoted by $!_X$. Similarly, if $0$ is initial in $\X$ then $?_X$ will denote the unique arrow $0\to X$. Moreover, when $\X$ is $\Set$ and $1$ is a singleton, $\delta_x$ will denote the arrow $1\to X$ with value $x\in X$.
	\item  $\mor(\X)$, $\mon(\X)$ and $\reg(\X)$ will denote the class of all arrows, monos and regular monos of $\X$, respectively.
	\item If $X$ is an object of a category $\X$, $\sub(X)$ will denote the collection of its \emph{subobjects}, i.e. equivalence classes of monos with codomain $X$ with respect to the relation which identifies $m: M\to X$ and $n: N\to X$ if and only if there exists an isomorphism $\phi: M\to N$ such that $n\circ \phi = m$. More generally, we will say that $[m]\leq [n]$ if and only if there exists an arrow $f:M\to N$ such that $n\circ f=m$. The resulting poset will be denoted by $(\sub(X), \leq)$.
	\item If $\mathcal{M}\subseteq \mon(\X)$ is a class of monomorphisms closed under composition with isomorphisms, we define $\msub{M}(X)$ to be the subclass of $\sub(X)$ given by classes of monos belonging to $\mathcal{M}$, $(\msub{M}(X), \leq)$ will denote the induced poset.
\end{itemize}

\subsection{The Van Kampen condition}

The key property that $\mathcal{M}, \mathcal{N}$-adhesive categories enjoy is given by  the so-called \emph{Van Kampen condition} \cite{brown1997van,johnstone2007quasitoposes,lack2005adhesive}. We will recall it and examine some of its consequences.

 \begin{defi} Let $\X$ be a category  and consider the two diagrams below
 	\[\xymatrix@C=10pt@R=10pt{&&&&&A'\ar[dd]|\hole_(.65){a}\ar[rr]^{g'} \ar[dl]_{m'} && B' \ar[dd]^{b} \ar[dl]_{n'} \\ A \ar[dd]_{m}\ar[rr]^{f}&&B\ar[dd]^{n}&&C'  \ar[dd]_{c}\ar[rr]^(.7){f'} & & D' \ar[dd]_(.3){d}\\&&&&&A\ar[rr]|\hole^(.65){g} \ar[dl]^{m} && B \ar[dl]^{n} \\C\ar[rr]_{g} &&D&&C \ar[rr]_{g} & & D}\]
	We say that the left square is a \emph{Van Kampen square} if:
	\begin{enumerate}
		\item it is a pushout square;
		\item 	in every commutative cube as the one above, with pullbacks as back and left faces, the top face is a pushout if and only if the front and right faces are pullbacks.
	\end{enumerate}

	Pushout squares which enjoy the ``if'' half of this condition are called \emph{stable}.
\end{defi}

\begin{rem}\label{rem:tec1} Van Kampen property and stability are invariant under isomorphism in two ways:
	\begin{itemize}
		\item 
		if a stable  (Van Kampen) stable pushout square of $m$ along $n$ exists, then every other pushout square of $m$ along $n$ is stable (Van Kampen);
		\item to prove that a given pushout square is stable (Van Kampen), it is enough to verify stability (the Van Kampen property) for a cube in which we have chosen representatives for the pullback squares given by the vertical (back and left) faces. 
	\end{itemize}
\end{rem}

 Before proceeding further, we need to recall a classical result about pullbacks \cite{mac2013categories,mclarty1992elementary}.

\begin{lem}\label{lem:pb1}
	Let $\X$ be a category, and consider the following diagram 	in which the right square is a pullback.
	\[\xymatrix{X \ar[d]_{a} \ar[r]^{f}& \ar[r]^{g} Y \ar[d]^{b}& Z \ar[d]^{c}\\ A \ar[r]_{h}& B \ar[r]_{k}& C}\]
	Then the whole rectangle is a pullback if and only if the left square is one.
\end{lem}

\begin{cor}\label{cor:cube} Let $\X$ be a category and let the solid part of the following cube be given
		\[\xymatrix@C=13pt@R=13pt{&Y'\ar[dd]|\hole_(.65){y}\ar[rr]^{g'} \ar@{.>}[dl]_{q'} && Z' \ar[dd]^{z} \ar[dl]_{r'} \\ B'  \ar[dd]_{b}\ar[rr]^(.65){k'} & & C' \ar[dd]_(.3){c}\\&Y\ar[rr]|\hole^(.65){g} \ar[dl]^{q} && Z \ar[dl]^{r} \\B \ar[rr]_{k} & & C}\]
	If the front face is a pullback then there is a unique $q':Y'\to B'$ completing the cube. If, moreover, the right and back faces are pullbacks, then the left face is a pullback too.
\end{cor}

We can dualize Lemma \ref{lem:pb1} to get half of the following. 

\begin{lem}\label{lem:po}
	Let $\X$ be a category, and consider the following diagram 	in which the left square is a pushout.
	\[\xymatrix@R=16pt{X \ar[d]_{p} \ar[r]^{f}& \ar[r]^{g} Y \ar[d]^{q}& Z \ar[d]^{r}\\ A \ar[r]_{h}& B \ar[r]_{k}& C}\]
	Then the whole rectangle is a pushout if and only if the right square is one. 
	
	Moreover, if $\X$ has pullbacks and the left square is stable, then stability of the rectangle is equivalent to that of the right square.
\end{lem}

We recall another property of Van Kampen squares from \cite{lack2005adhesive}.
\begin{prop}\label{prop:pbpo} Let $m:A\to C$ be a monomorphism in a category $\X$, then every Van Kampen square
		\[\xymatrix@R=16pt{A\ar[r]^{g} \ar[d]_{m} & B \ar[d]^{n} \\ C \ar[r]_{f}  & D}\]
		is also a pullback square and $n$ is a monomorphism.
\end{prop}

Finally, some kind of left cancellation property holds for pullbacks \cite[Lemma $2.2$]{garner2012axioms}.

\begin{lem}\label{lem:pb2}
	Let $\X$ be a category with pullbacks, given the following diagrams:
	\[
	\xymatrix@R=16pt{Y\ar[r]^{f_2} \ar[d]_{f_1} & X_2 \ar[d]^{r_2} & Z_1 \ar[d]_{x_1}\ar[r]^{z_1} & W \ar[r]^{w} \ar[d]_{r} & Q'\ar[d]^{q} & Z_2 \ar[d]_{x_2} \ar[r]^{z_2}  & W  \ar[r]^{w} \ar[d]_{r}  & Q' \ar[d]^{q}\\ X_1 \ar[r]_{r_1} &R  & X_1 \ar[r]_{r_1} & R \ar[r]_{s}  & Q& X_2 \ar[r]_{r_2} & R \ar[r]_{s} & Q}\]
	if the first square is a stable pushout and the whole rectangles and their left halves are pullbacks, then their common right half is a pullback too.	
\end{lem}

\subsection{Definition of $\mathcal{M}, \mathcal{N}$-adhesivity}

In this section we will define the notion of $\mathcal{M}, \mathcal{N}$-adhesivity and explore some of the consequences of such a property. Let us start fixing some terminology.

\begin{defi}
Let $\X$ be a category and $\mathcal{A}$, $\mathcal{B}$ two classes of arrows, we say that  $\mathcal{A}$ is
\begin{itemize}
	\item 
	\emph{stable under pushouts (pullbacks)} if for every pushout (pullbacks) square 
	\[\xymatrix@R=16pt{A\ar[r]^f  \ar[d]_{m}& B \ar[d]^n \\ C \ar[r]_g & D}\]
 if $m \in \mathcal{A}$ ($n\in \mathcal{A}$) then $n \in \mathcal{A}$ ($m \in \mathcal{A}$);
	\item \emph{closed under composition} if $g, f\in \mathcal{A}$ implies $g\circ f\in \mathcal{A}$ whenever $g$ and $f$ are composable;
	\item \emph{closed under $\mathcal{B}$-decomposition} if $g\circ f\in \mathcal{A}$ and $g\in \mathcal{B}$ implies $f\in \mathcal{A}$;
	\item \emph{closed under decomposition} if it is closed under $\mathcal{A}$-decomposition.
\end{itemize}
\end{defi}

\begin{rem}Clearly, ``decomposition'' corresponds to ``left cancellation'', but we prefer to stick to the name commonly used in literature (see e.g.~\cite{habel2012mathcal}).
\end{rem}

\begin{rem}\label{rem:stab1}	
	In the literature, when studying \emph{stable systems of monos}, stability under pullbacks (pushouts) is frequently assumed to imply the existence of those (co)limits. In this paper we will not make this assumption, because it would lead to a definition of $\mathcal{M}, \mathcal{N}$-adhesivity more restrictive than the one found in the literature (see Remark \ref{rem:stab2}).
\end{rem}

\begin{exa}\label{ex:split}
	In every category $\X$, split monomorphisms (i.e. those arrows which have a left inverse) are stable under pushouts. Indeed, take a pushout square 
		\[\xymatrix@R=16pt{A\ar[r]^f  \ar[d]_{m}& B \ar[d]^n \\ C \ar[r]_g & D}\] with $m$ a split monomorphism. Let $r:C\to A$ be a left inverse of $m$, then
	\begin{align*}
	f\circ r\circ m = f\circ \id{A}=f=\id{B}\circ f 
	\end{align*}
 This equality in turn entails the existence of a unique $t:D\to B$ as in the diagram below.
	\[\xymatrix@R=16pt{A\ar[r]^f  \ar[d]_{m}& B \ar@/^.3cm/[ddr]^{\id{B}} \ar[d]^n \\ C \ar@/_.3cm/[dr]_{r} \ar[r]_g & D \ar@{.>}[dr]^{t}\\ & A \ar[r]_{f}&B}\]
\end{exa}

\begin{lem}\label{lem:stab}Let $\mathcal{M}$ be a class of monos in a category $\X$ which is stable under pullbacks and contains all  split monos. If pushouts along arrows in $\mathcal{M}$ exist and are Van Kampen, then $\mathcal{M}$ is stable under pushouts. 
\end{lem}

We are now ready to give the definition of $\mathcal{M},\mathcal{N}$-adhesive category	
\begin{defiC}[\cite{habel2012mathcal,peuser2016composition}]\label{def:class}
	Let $\X$ be a category, $\mathcal{M}\subseteq \mon(\X)$ and $\mathcal{N}\subseteq\mor(\X)$, we say that the pair $(\mathcal{M}, \mathcal{N})$ is a \emph{preadhesive structure} on $\X$ if the following conditions hold.
	\begin{enumerate}
		\item $\mathcal{M}$ and $\mathcal{N}$ contain all isomorphisms and are closed under composition and decomposition;
		\item $\mathcal{N}$ is closed under $\mathcal{M}$-decomposition;
		\item $\mathcal{M}$ and $\mathcal{N}$ are stable under pullbacks and pushouts.
	\end{enumerate}
	Given a preadhesive structure $(\mathcal{M}, \mathcal{N})$, we say that $\X$ is \emph{$\mathcal{M}, \mathcal{N}$-adhesive} if
	\begin{enumerate}
		\item for every $m:X\to Y$ in $\mathcal{M}$ and $g:Z\to Y$, a pullback square
		\[\xymatrix@R=16pt{P\ar[r]^p \ar[d]_{n}& X \ar[d]^{m}\\ Z \ar[r]_g& Y}\]
		exists, such pullbacks will be called \emph{$\mathcal{M}$-pullbacks};
		\item for every $m:X\to Y$ in $\mathcal{M}$ and $n:X\to Z$ in $\mathcal{N}$, a pushout square
		\[\xymatrix@R=16pt{X \ar[r]^n \ar[d]_{m}& Z \ar[d]^{q}\\ Y\ar[r]_p &Q}\]
		exists, such pushouts  will be called \emph{$\mathcal{M}, \mathcal{N}$-pushouts}; 
		\item  $\mathcal{M}, \mathcal{N}$-pushouts are Van Kampen squares.
	\end{enumerate}
\end{defiC}

\begin{rem}\label{rem:diff}
	Our notion of $\mathcal{M}$, $\mathcal{N}$-adhesivity is different from the one of \cite{habel2012mathcal}: in that paper,  $\mathcal{M}, \mathcal{N}$-pushouts are required to satisfy the Van Kampen condition only for cubes in which the vertical arrows belong to $\mathcal{M}$. 
\end{rem}

\begin{rem}\label{rem:stab2}
	As noted in Remark \ref{rem:stab1}, we do not assume that stability under pullbacks or pushouts of $\mathcal{M}$ and $\mathcal{N}$ implies the existence of those (co)limits. Indeed, if we were to make this assumption, an $\mathcal{M}, \mor(\X)$-adhesive category would have all pushouts, but this does not correspond to the definition of $\mathcal{M}$-adhesivity in the literature \cite{azzi2019essence} (see also \Cref{subsec:m-ad}).
\end{rem}

Proposition \ref{prop:pbpo} yields at once the following fact.
\begin{prop}\label{prop:pbpoad}
	If $\X$ is an $\mathcal{M},\mathcal{N}$-adhesive category, then $\mathcal{M}, \mathcal{N}$-pushouts are also pullback squares.
\end{prop}

It can be shown that the class of $\mathcal{M}, \mathcal{N}$-adhesive categories is closed under the most common categorical constructions \cite{CastelnovoGM22,Castelnovo2024simple,castelnovo2023thesis}.

\begin{lem}\label{cor:varie1}
	Let $\{\X_i\}_{i\in I}$ be a family of categories such that each $\X_i$ is $\mathcal{M}_i,\mathcal{N}_i$-adhesive. Then the product category $\prod_{i\in I}\X_i$ is $\prod_{i\in I}\mathcal{M}_i,\prod_{i\in I}\mathcal{N}_i$-adhesive, where
	\begin{align*}
		\textstyle \prod_{i\in I}\mathcal{M}_i&:= \textstyle \{(m_i)_{i\in I}\in \mathsf{Mor}(\prod_{i\in I}\X_i)\mid m_i \in \mathcal{M}_i \text{ for every } i \in I\}\\
		\textstyle \prod_{i\in I}\mathcal{N}_i&:= \textstyle \{(n_i)_{i\in I}\in \mathsf{Mor}(\prod_{i\in I}\X_i)\mid n_i \in \mathcal{N}_i \text{ for every } i \in I\}
	\end{align*}
\end{lem}
\begin{lem}\label{cor:varie2}
	Let $\X$ be an $\mathcal{M},\mathcal{N}$-adhesive category with all pullbacks and pushouts. Then for every other category $\Y$, the category of functors $\X^{\Y}$ is $\mathcal{M}^{\Y},\mathcal{N}^{\Y}$-adhesive, where
	\begin{align*}
		\mathcal{M}^{\Y}&:=\{\eta \in \arr{X^Y} \mid \eta_Y \in \mathcal{M} \text{ for every object } Y \text{ of } \Y \}\\
		\mathcal{N}^{\Y}&:=\{\eta \in \arr{X^Y} \mid \eta_Y \in \mathcal{N} \text{ for every object } Y \text{ of } \Y \}
	\end{align*} 
\end{lem}

\begin{lem}\label{cor:varie3} Let $\catname{Y}$ be a full subcategory of an $\mathcal{M}, \mathcal{N}$-adhesive category $\catname{X}$ with all pullbacks. If $(\mathcal{M}', \mathcal{N}')$ is a preadhesive structure on $\Y$ such that $\mathcal{M}'\subseteq \mathcal{M}$, $\mathcal{N}'\subseteq \mathcal{N}$ and $\catname{Y}$ is closed in $\catname{X}$ under pullbacks and $\mathcal{M'}, \mathcal{N'}$-pushouts, then $\catname{Y}$ is $\mathcal{M}', \mathcal{N'}$-adhesive.
\end{lem}

\begin{lem}\label{comma} 
	Let $\catname{X}$ and $\catname{Y}$ be respectively $\mathcal{M},\mathcal{N}$-adhesive and $\mathcal{M}',\mathcal{N}'$-adhesive categories with all pushouts and pullbacks. Let also   $L:\catname{X}\rightarrow \catname{Z}$ be a functor that preserves $\mathcal{M}, \mathcal{N}$-pushouts, and  $R:\catname{Y}\rightarrow \catname{Z}$ a functor which preserves pullbacks. Then $\comma{L}{R}$ is $\cma{M}{M'}, \cma{N}{N'}$-adhesive, where 
	\begin{align*}
		\cma{M}{M}'&:=\{(h,k)\in \mathsf{Mor}(\comma{L}{R}) \mid h\in \mathcal{M}, k\in \mathcal{M}'\}\\
		\cma{N}{N}&:=\{(h,k)\in \mathsf{Mor}(\comma{L}{R}) \mid h\in \mathcal{N}, k\in \mathcal{N}'\}
	\end{align*}
\end{lem}

These three results, in turn, allow us to provide examples of graphical and hypergraphical $\mathcal{M}, \mathcal{N}$-adhesive categories (see again \cite{CastelnovoGM22,Castelnovo2024simple} for details and other examples).
\begin{exa} Take the category $\gr$ of graphs and let $\dgr$ and $\dg$ be its full subcategories of given by simple graphs and directed acyclic graphs.  We say that a monomorphism $(f,g):\mathcal{{G}}\rightarrow \mathcal{{H}}$ in $\gr$ \emph{downward closed} if, for all $e\in E_\mathcal{{H}}$,  $t_\mathcal{H}(e) \in g(V_\mathcal{G})$ implies $e \in f(E_\mathcal{G})$. We denote by $\rta$, $\rt$ and $\rtd$ the classes of downward closed morphisms in $\gr$, $\dgr$ and $\dg$ respectively. Then we have:
	\begin{itemize}
		\item $\dgr$  is $\reg(\dgr), \mon(\dgr)$-adhesive and $\mon(\dgr),\reg(\dgr)$-adhesive;
		\item $\dg$ is $\rtd, \mon(\dg)$-adhesive.
	\end{itemize}
\end{exa}

\begin{exa}\label{ex:1}Let  $\catname{Tree}$ be the category of \emph{tree orders}, i.e. partial orders $(E, \leq)$ such that for every $e\in E$, the set
	$\pred{e}:=\{e'\in E \mid e' \leq e\}$	is  finite and totally ordered by the restriction of $\leq$.  The category $\catname{HGraph}$ of \emph{hierarchical graphs} \cite{palacz2004algebraic} is defined as the comma category $\comma{L}{R}$ where 
	$L:\catname{Tree}\to \Set$ is the forgetful functor and $R:\Set\to \Set$ is the functor sending $X\mapsto X\times X$.
	Then, $\catname{HGraph}$ is adhesive. 
\end{exa}

\begin{exa}\label{ex:2} More generally, if we want to use simple or directed acyclic graphs to model the hierarchy between edges, we can define the categories of \emph{$\dgr$-graphs} respectively \emph{$\dg$-graphs} as the comma categories $\comma{L}{R}$ where
	\begin{itemize}
	\item $L$ is one of the two forgetful functors $U_\dgr: \dgr \to \Set$ or $U_\dg:\dg\to \Set$;
	\item $R:\Set\to \Set$ is the functor sending $X\mapsto X\times X$.
	\end{itemize}
Then we get the following results:
\begin{itemize}
	\item 	 the categories of $\dgr$-graphs  is $\mathcal{M}, \mathcal{N}$-adhesive, where
	\begin{align*}
		\mathcal{M}&:=\{((h_1,h_2), k)\in \comma {U_\dgr}{R}\mid (h_1,h_2)\in \reg(\dgr), k \in \mon(Set)\}\\
		\mathcal{N}&:=\{((h_1,h_2), k)\in \comma {U_\dgr}{R}\mid (h_1,h_2)\in \mon(\dgr)\}
	\end{align*}

	\item the categories of $\dg$-graphs is $\mathcal{M}, \mathcal{N}$-adhesive where	\begin{align*}
		\mathcal{M}&:=\{((h_1,h_2), k)\in \comma {U_\dg}{R}\mid (h_1,h_2)\in \rtd, k \in \mon(\Set)\}\\
		\mathcal{N}&:=\{((h_1,h_2), k)\in \comma {U_\dg}{R}\mid (h_1,h_2)\in \mon(\dg)\}
	\end{align*}
\end{itemize}
\end{exa}

\begin{exa}We can further modify Examples \ref{ex:1} and \ref{ex:2} to get hypergraphical versions of it: it is enough to replace the functor $\Set\to \Set$ sending $X$ to $X \times X$, with the functor sending $X$ to $X^\star \times X^\star$, where $(-)^\star:\Set \to \Set$ is the Kleene star. 
\end{exa}

\subsection{Relation with $\mathcal{M}$-adhesivity}\label{subsec:m-ad}
We will end this section proving that, under suitable hypothesis, $\mathcal{M}, \mathcal{N}$-adhesivity subsumes  
$\mathcal{M}$-adhesivity as defined in \cite{azzi2019essence}.

\begin{defi}
 Let $\X$ be a category, a \emph{stable system of monos}  $\mathcal{M}$ is a class of monomorphisms  closed under composition, containing all isomorphisms and stable under pullbacks.
\end{defi}

\begin{lem}\label{lem:deco}Let $f:X\to Y$ be an arrow in a category $\X$ equipped of a stable system of monos  $\mathcal{M}$. For any mono $m:Y\to Z$, if $m\circ f \in\mathcal{M}$ then $f\in \mathcal{M}$.
\end{lem}

\begin{defiC} [\cite{azzi2019essence}] Given a stable system of monos $\mathcal{M}$ on a category $\X$, we say that $\X$ is \emph{$\mathcal{M}$-adhesive} if 
	\begin{enumerate}
		\item 	it has $\mathcal{M}$-pullbacks;
		\item  for every $m:X\to Y$ in $\mathcal{M}$ and for any other arrow $f:X\to Z$, a pushout square 
		\[\xymatrix@R=16pt{X \ar[r]^{f} \ar[d]_{m}& Z\ar[d]^{n}\\ Y \ar[r]^{g} & Q}\]
		exists and it is a Van Kampen square. 
	\end{enumerate}
\end{defiC}
\begin{rem}\label{rem:diff2}  We will stick to the notion of
	$\mathcal{M}$-adhesivity as defined in \cite{azzi2019essence}, as noted in Remark \ref{rem:diff}. 
	Other authors have introduced weaker notions of $\mathcal{M}$-adhesivity, where the Van Kampen condition is required to hold only for some cubes; see, e.g.,  \cite{behr2022fundamentals,ehrig2006fundamentals,ehrig2014adhesive,ehrig2004adhesive,sobocinski2020rule}, where our $\mathcal{M}$-adhesive categories are called \emph{adhesive HLR categories}.
	
	On the other hand,  in \cite{azzi2019essence} no requirement about the existence of pullbacks or $\mathcal{M}$-pullbacks is made, while in \cite{garner2012axioms,johnstone2007quasitoposes,lack2005adhesive} adhesive and quasiadhesive categories are required to have all pullbacks. Mimicking the definition of $\mathcal{M}, \mathcal{N}$-adhesivity,  for us an $\mathcal{M}$-adhesive category  must have $\mathcal{M}$-pullbacks. 
\end{rem}

\begin{prop}\label{prop:po}
	Let $\X$ be an $\mathcal{M}$-adhesive category and suppose that every split mono is in  $\mathcal{M}$, then $\mathcal{M}$ is stable under pushouts.
\end{prop}

\begin{exa}\label{ex:ad}The first, and fundamental, example is when $\mathcal{M}$ is the class of all monomorphisms: in this case $\mathcal{M}$-adhesivity is simply called \emph{adhesivity}.
\end{exa}

One would weaken the previous example using regular monos instead of simple monomorphisms. The problem is that $\reg(\X)$ is not in general closed under composition. Anyway, the Van Kampen property and the existence of $\reg(\X)$-pushouts ensure the closure of regular monos under composition.

\begin{prop}\label{prop:qad}
	For $\X$ a category with $\reg(\X)$-pullbacks, the following are equivalent:
	\begin{enumerate}
		\item $\reg(\X)$ is a stable system of monos and $\X$ is $\reg(\X)$-adhesive;
		\item pushouts along regular monos exist and are Van Kampen.
	\end{enumerate}
\end{prop} 
%

\begin{rem}
A category with pullbacks and pushouts along regular monos and in which such pushouts are Van Kampen is what in the literature is usually called a \emph{quasiadhesive category}, a notable exception is \cite{garner2012axioms}, in which \emph{rm-adhesive} is used. 
\end{rem}
\begin{lem}\label{lem:mad}
	Let $\mathcal{M}$ be a stable system of monos in a category $\X$ which is also stable under pushouts, then the following are equivalent: 
		\begin{enumerate}
		\item $\X$ is $\mathcal{M}$-adhesive;
		\item  $\X$ is $\mathcal{M}, \mor(\X)$-adhesive.
	\end{enumerate}
\end{lem}
%

We can apply Proposition \ref{prop:po} to get the following corollary at once.

\begin{cor}\label{cor:mad}	
	Let $\mathcal{M}$ be a stable system of monos in a category $\X$ and suppose that it contains all split monomorphisms., then the following are equivalent: 
	\begin{enumerate}
		\item $\X$ is $\mathcal{M}$-adhesive;
		\item  $\X$ is $\mathcal{M}, \mor(\X)$-adhesive.
	\end{enumerate}
\end{cor}

If we specialize the previous results to the classes of monos and regular monos we get the following.
\begin{cor}\label{cor:equi}
	A category $\X$ is adhesive if and only if it is $\mon(\X), \mor(\X)$-adhesive and it is quasiadhesive if and only if it is $\reg(\X), \mor(\X)$-adhesive).
\end{cor}

 
\section{$\mathcal{N}$-(pre)adhesive morphisms} \label{sec:n-pread-mor}
 In \cite{garner2012axioms,johnstone2007quasitoposes} a criterion to establish quasiadhesivity is provided, involving the closure of regular monos under unions. The aim of this section is to adapt those results to the setting of $\mathcal{M}, \mathcal{N}$-adhesivity.

The first step that we need to take is to generalize the notions of \emph{adhesive} and \emph{preadhesive morphism} provided in \cite{garner2012axioms}.

\begin{defi}\label{def:ade}
	Given a class $\mathcal{N}$ of arrows of a category $\X$, we say that $\mathcal{N}$ is a \emph{matching class} if
	\begin{enumerate}
		\item it contains all isomorphisms;
		\item is closed under composition and decomposition;
		\item is stable under pullbacks and pushouts.
	\end{enumerate}

Given a matching class $\mathcal{N}$,  a morphism $m:X \to Y$  in $\X$ is \emph{$\mathcal{N}$-preadhesive} if for every $n:X\to Z$ in $\mathcal{N}$, a stable pushout square
	\[\xymatrix@R=16pt{X\ar[r]^{n} \ar[d]_{m} & Z \ar[d]^{p} \\ Y \ar[r]_{q}  & W}\]	
	exists and it is also a pullback for $p$ along $q$. The arrow $m$ will be called \emph{$\mathcal{N}$-adhesive} if all its pullbacks are $\mathcal{N}$-preadhesive.

	We will denote by $\pad{N}$ and by $\ad{N}$ the classes of, respectively, $\mathcal{N}$-preadhesive and $\mathcal{N}$-adhesive morphisms. 
\end{defi}
\begin{nota}Instead of ``$\mor(\X)$-(pre)adhesive'' we will use ``(pre)adhesive''. 
\end{nota}

\begin{exa}\label{ex:ade}
	If $\X$ is an $\mathcal{M}, \mathcal{N}$-adhesive category then $\mathcal{N}$ is a matching class. Moreover, $\mathcal{M}, \mathcal{N}$-pushouts are Van Kampen squares, so every $m\in \mathcal{M}$ is preadhesive. Since $\mathcal{M}$ is closed under pullback this implies that every arrow in $\mathcal{M}$ is also adhesive.
\end{exa}

The following proposition collects some useful facts about $\mathcal{N}$-(pre)adhesive morphisms.
\begin{prop}\label{prop:dec}Let $\mathcal{N}$ be a matching class on a category $\X$, then the following hold true:
	\begin{enumerate}
		\item if $m$ is $\mathcal{N}$-adhesive then it is $\mathcal{N}$-preadhesive;
		\item every isomorphism is $\mathcal{N}$-adhesive;
		\item if $n\in \mathcal{N}$ is $\mathcal{N}$-preadhesive then it is a regular mono;
		\item the class $\pad{N}$ is closed under composition;
		\item $\ad{N}$ is stable under pullbacks;
		\item if $\X$ has pullbacks along $\mathcal{N}$-adhesive arrows, then $\ad{N}$ is closed under composition.
	\end{enumerate}
\end{prop}

\begin{cor}\label{cor:rega}
	In any category $\X$, $\mor(\X)_{\mathsf{a}}\subseteq \reg(\X)$.
\end{cor}

\begin{cor}\label{cor:stru}
Let $\mathcal{N} $ be a matching class on a category $\X$ with pullbacks, then:
\begin{enumerate}
	\item $\ad{N}\cap \mon(\X)$ is a stable system of monos;
	\item if $\ad{N}\cap \mon(\X)$ is stable under pushouts, then $(\ad{N}\cap \mon(\X), \mathcal{N})$ is a preadhesive structure.
\end{enumerate}
\end{cor}

In general we cannot guarantee closure of $\ad{N}$ under all pushouts, nonetheless we can still establish some result  along this line. 

\begin{lem}[Cfr.~\cite{garner2012axioms}, Prop.~2.3]\label{lem:adpush} Let $\mathcal{N}$ be a matching class in a category $\X$ with pullbacks and consider the following pushout 
	\[\xymatrix@R=16pt{X\ar[r]^{n} \ar[d]_{m_1} & Z \ar[d]^{m_2}\\
	Y \ar[r]_{g}& W}\]
with $n\in \mathcal{N}$. If $m_1$ is mono and $\mathcal{N}$-adhesive, then:
\begin{enumerate}
	\item $m_2$ is mono;
	\item $m_2$ is $\mathcal{N}$-preadhesive;
	\item $m_2$ is $\mathcal{N}$-adhesive.
\end{enumerate}
\end{lem}
\begin{proof}
	\begin{enumerate}[leftmargin=0pt,itemindent=1.7em]
		\item  Since $\X$ has pullbacks, we have a diagram
		\[\xymatrix@R=16pt{X\ar@/^.3cm/[drr]^{n} \ar@/_.3cm/[ddr]_{n} \ar@{.>}[dr]^{h}\\ &P\ar[r]^{p} \ar[d]_{q} & Z \ar[d]^{m_2}\\&
			Z \ar[r]_{m_2}& W}\]
		in which the square is a pullback, so that the dotted $h$ exists because of its universal property. We aim to prove that $p$ and $q$ are isomorphisms.
		To get the thesis it is now enough to prove that $p$ is an isomorphism. Indeed, when taking the pullback of an arrow along itself, the fact that one projection is an isomorphism entails that also the other one is so. Therefore, if $p$ is invertible then $q$ is an isomorphism too and thus $m_2$ a monomorphism.
		
		To this end, we can build the following cube, starting from the bottom face (the hypothesis) and the right face (the pullback above):
		\[\xymatrix@C=13pt@R=13pt{&X\ar[dd]|\hole_(.65){\id{X}}\ar[rr]^{h} \ar[dl]_{\id{X}} && P \ar[dd]^{q} \ar[dl]_{p} \\ X  \ar[dd]_{m_1}\ar[rr]^(.65){n} & & Z \ar[dd]_(.3){m_2} \\&X\ar[rr]|\hole^(.65){n} \ar[dl]^{m_1} && Z \ar[dl]^{m_2}\\Y \ar[rr]_{g} & & W }\]
		By Proposition \ref{prop:dec}(1) the bottom and front faces are stable pushouts and pullbacks because $m_1$ is $\mathcal{N}$-adhesive, and the right square is a pullback by  hypothesis.  Lemma \ref{lem:pb1} entails that the rectangle
		\[\xymatrix{X \ar@/^.4cm/[rr]^{n}\ar[r]_{h} \ar[d]_{\id{X}}&  P \ar[r]_{p}\ar[d]_{q} & Z \ar[d]^{m_2} \\ X \ar@/_.4cm/[rr]_{g\circ m_1}\ar[r]^{n} & Z \ar[r]^{m_2} & W }\]
		are pullbacks, thus the back faces of the two cubes are pullbacks too.
		The left faces are pullbacks because $m_1$ is mono, therefore, by stability of the bottom faces, it follows that the top faces pushouts and thus $p$ is an isomorphism.
		\item Let $k:Z\to Q$ be another arrow in $\mathcal{N}$ and consider the diagram
		\[\xymatrix@R=16pt{X \ar[r]^{n} \ar[d]_{m_1}&  Z \ar[r]^{k}\ar[d]^{m_2} & Q \ar[d]^{f} \\ Y \ar@/_.4cm/[rr]_{t}\ar[r]^{g} & W \ar@{.>}[r]^{s} & S}\]
		in which the left square and the external rectangle are stable pushouts and pullbacks. Since
		\[f\circ k \circ n = t\circ m_1\]
		 the universal property of the left square yields the dotted $s$.  By Lemma \ref{lem:po} the square so obtained is a stable pushout,  thus we are left with showing that it is a pullback.
		 Given the solid part of the diagram
		 \[\xymatrix@R=16pt{L \ar@{.>}[dr]^{l_2}\ar@/^.3cm/[drrr]^{l_1} \ar@/_.3cm/[ddr]_{l_2}\\&Z \ar[rr]^{k} \ar[d]^{\id{Z}}&  & Q \ar[d]^{f} \\ & Z \ar[r]_{m_2} & W \ar[r]_{s} & S}\]
		 we have
		 \begin{align*}
		 	f\circ l_1&=s\circ m_2\circ l_2=f\circ k\circ l_2
		 \end{align*}
		 By the previous point $f$ is mono and thus the following rectangle is a pullback
\[\xymatrix@R=16pt{Z \ar[rr]^{k} \ar[d]_{\id{Z}}&  & Q \ar[d]^{f} \\ Z \ar[r]_{m_2} & W \ar[r]_{s} & S}\]
	The thesis now follows applying the previous point and Lemma \ref{lem:pb2} to the following diagrams.
			\[
		\xymatrix@R=16pt{X\ar[r]^{n} \ar[d]_{m_1} & Z \ar[d]^{m_2} & X \ar[d]_{m_1}\ar[r]^{n} & Z \ar[r]^{k} \ar[d]_{m_2} & Q\ar[d]^{f} & Z \ar[d]_{\id{Z}} \ar[r]^{\id{Z}}  & Z  \ar[r]^{k} \ar[d]_{m_2}  & Q \ar[d]^{f}\\ Y \ar[r]_{g} &W  & Y \ar[r]_{g} &W \ar[r]_{s}  & S& Z \ar[r]_{m_2} & W \ar[r]_{s} & S}\] 
		\item  Take an arrow $w:W'\to W$ and consider the following cube, in which the solid faces are pullbacks
		\[\xymatrix@C=13pt@R=13pt{&X'\ar[dd]|\hole_(.65){x}\ar[rr]^{n'} \ar@{.>}[dl]_{m'_1} && Z' \ar[dd]^{z} \ar[dl]_{m'_2} \\ Y'  \ar[dd]_{y}\ar[rr]^(.65){g'} & & W' \ar[dd]_(.3){w}\\&X\ar[rr]|\hole^(.65){n} \ar[dl]^{m_1} && Z \ar[dl]^{m_2} \\Y\ar[rr]_{g} & & W}\]
By Corollary \ref{cor:cube} the arrow $m'_1:X'\to Y'$ exists and the added face is a pullback. Since the bottom face is a stable pushout then the top face is a pushout too. By point $5$ of $\ref{prop:dec}$ $m'_1$ is $\mathcal{N}$-adhesive and, since $\mathcal{N}$ is matching, $n'$ is in $\mathcal{N}$. The previous point of this lemma implies that $m'_2$ is $\mathcal{N}$-preadhesive and we can conclude.		\qedhere 
	\end{enumerate}
\end{proof}

\begin{cor}\label{cor:pads} If $\X$ is a category with pullbacks then $(\mor(\X)_{\mathsf{a}}, \mor(\X))$ is a preadhesive structure.
\end{cor}

Finally, $\mathcal{N}$-adhesivity allows us to compute suprema of certain pairs of subobjects.

\begin{prop}[Cfr.~\cite{garner2012axioms}, Prop.~2.4]\label{prop:uni}
Let $\mathcal{N}$ be a matching class in a category $\X$ with pullbacks. Given an $\mathcal{N}$-adhesive mono $m:M\to X$  and another mono $n:N\to X$ in $\mathcal{N}$, consider the diagram
	\[\xymatrix@R=16pt{P \ar[r]^{p_1} \ar[d]_{p_2} & M\ar[d]^{u_2} \ar@/^.3cm/[ddr]^{m} \\ N \ar@/_.3cm/[drr]_{n}\ar[r]_{u_1} & U \ar@{.>}[dr]^{u} \\ && X}\]
	in which the outer boundary form a pullback and the inner square a pushout. Then the dotted arrow $u:U\to X$ is a monomorphism and, in $(\sub(X), \leq)$
	\[[u]=[m]\vee [n]\]  
\end{prop}
\begin{rem}
	Notice that the  $p_2$ and $p_1$ are both monos, moreover, $p_2$ is $\mathcal{N}$-preadhesive while $p_1\in\mathcal{N}$, as the pullback of $n$, thus the inner pushout exists.
\end{rem}
\begin{proof}[Proof of Proposition \ref{prop:uni}.]
	Consider the following two pullback squares
	\[\xymatrix@R=16pt{Q \ar[r]^{q_1} \ar[d]_{q_2} & N\ar[d]^{n} & W \ar[r]^{w_1} \ar[d]_{w_2}& M\ar[d]^{m}\\
		U \ar[r]_{u}& X & U \ar[r]_{u}& X}\]

Since $n\circ p_2$ is equal to $m\circ p_1$, we get arrows $f_1:N\to Q$, $f_2:P\to Q$, $g_1:M\to W$, $g_2:P\to W$ making the following diagrams commute.
	\[\xymatrix@C=20pt@R=16pt{N\ar@/^.4cm/[rr]^{\id{N}}\ar[r]_{f_1}  \ar[d]_{\id{N}} &Q\ar[r]_{q_1} \ar[d]_{q_2} & N \ar[d]^{n}  &P\ar@/^.4cm/[rr]^{p_2}\ar[r]_{f_2}  \ar[d]_{p_1} &Q\ar[r]_{q_1} \ar[d]_{q_2} & N \ar[d]^{n} & M\ar@/^.4cm/[rr]^{\id{M}}\ar[r]_{g_1}  \ar[d]_{\id{M}} &W\ar[r]_{w_1} \ar[d]_{w_2} & M \ar[d]^{m} & P\ar@/^.4cm/[rr]^{p_1}\ar[r]_{g_2}  \ar[d]_{p_2} &W\ar[r]_{w_1} \ar[d]_{w_2} & M \ar[d]^{m}
		\\ N \ar@/_.4cm/[rr]_{n}\ar[r]^{u_1} & U \ar[r]^{u}  & X  & M\ar@/_.4cm/[rr]_{m}\ar[r]^{u_2}   &U\ar[r]^{u}  & X &  M \ar@/_.4cm/[rr]_{m}\ar[r]^{u_2} & U \ar[r]^{u}  & X & N\ar@/_.4cm/[rr]_{n}\ar[r]^{u_1}  &U\ar[r]^{u} & X }\]
	Their outer edges are pullbacks, thus in the following cubes, the vertical faces are pullbacks
	\[\xymatrix@C=13pt@R=13pt{&P\ar[dd]|\hole_(.65){\id{P}}\ar[rr]^{\id{P}} \ar[dl]_{p_2} && P \ar[dd]^{p_1} \ar[dl]_{f_2} & & P\ar[dd]|\hole_(.65){\id{P}}\ar[rr]^{p_1} \ar[dl]_{\id{P}} && M \ar[dd]^{\id{M}} \ar[dl]_{g_1}\\ N  \ar[dd]_{\id{N}}\ar[rr]^(.65){f_1} & & Q \ar[dd]_(.3){q_2}&& P  \ar[dd]_{p_2}\ar[rr]^(.65){g_2} & & W \ar[dd]_(.3){w_2}\\&P\ar[rr]|\hole^(.65){p_1} \ar[dl]^{p_2} && M \ar[dl]^{u_2} && P\ar[rr]|\hole^(.65){p_1} \ar[dl]^{p_2} && M \ar[dl]^{u_2}\\N \ar[rr]_{u_1} & & U && N \ar[rr]_{u_1} & & U}\]
	$p_2$ is $\mathcal{N}$-preadhesive, so the top faces are pushouts and therefore $f_1$, and $g_1$ are isomorphisms with inverses given by $q_1$ and $w_1$.  But then, since
	$u_1=q_2\circ f_1$ and  $u_2=w_2\circ g_1$,	we can further deduce that the squares below  are both pullbacks.
	\[\xymatrix@R=16pt{N \ar[r]^{\id{N}} \ar[d]_{u_1} & N\ar[d]^{n} & M \ar[r]^{\id{M}} \ar[d]_{u_2}& M\ar[d]^{m}\\
		U \ar[r]_{u}& X & U \ar[r]_{u}& X}\]
	We have three diagrams 
	\[\xymatrix@R=16pt{P\ar[d]_{p_2} \ar[r]^{p_1} & M \ar[d]^{m}  & N \ar[d]_{\id{N}} \ar[r]^{u_1}& U\ar[d]_{\id{U}} \ar[r]^{\id{U}}  & U \ar[d]^{u} &  M \ar[d]_{\id{M}} \ar[r]^{u_2}& U \ar[r]^{\id{U}} \ar[d]_{\id{U}}  & U \ar[d]^{u}\\N \ar[r]_{n} & X & N \ar@/_.3cm/[rr]_{n}\ar[r]^{u_1} & U  \ar[r]^u  &X & M\ar@/_.3cm/[rr]_{m}\ar[r]^{u_2} & U  \ar[r]^u  &X }\]
	and we have just proved that the rectangles are pullbacks, thus we can apply Lemma \ref{lem:pb2} to deduce that the following square is a pullback---which means exactly that $u$ is a mono.
	\[\xymatrix@R=16pt{ U \ar[r]^{\id{U}} \ar[d]_{\id{U}}  & U \ar[d]^{u}\\ U  \ar[r]^u  &X }\]
	
	For the second half: suppose that $k:K\to X$ is an upper bound for $m$ and $n$, thus there exists $k_1:M\to K$ and $k_2:N\to K$ such that $m=k\circ k_1$ and $n=k\circ k_2$ so that
	\begin{align*}
		k\circ k_1 \circ p_1 &= m\circ p_1= n\circ p_2= k\circ k_2\circ p_2
	\end{align*}
	$k$ is mono, thus this implies that there exists a unique $h:U\to K$ such that 
	\[k_2=h\circ u_1 \qquad k_1=h\circ u_2\]
	and we have
	\begin{gather*}
		k\circ h \circ u_1=k\circ k_2=n =u \circ u_1 \qquad
		k\circ h \circ u_2=k\circ k_1=m =u \circ u_2 
	\end{gather*}
	showing that $u=k\circ h$, i.e. $u\leq k$.
\end{proof}

\section{From $\mathcal{M}, \mathcal{N}$-unions to $\mathcal{M}, \mathcal{N}$-adhesivity}\label{sec:uniad}

Given a preadhesive structure $(\mathcal{M}, \mathcal{N})$ such that $\mathcal{M}\subseteq \ad{N}$, in this section we will show how to deduce $\mathcal{M}, \mathcal{N}$-adhesivity  from the closure of $\mathcal{M}$ under some kind of unions.

\begin{defi} Let $(\mathcal{M}, \mathcal{N})$ be a preadhesive structure, a monomorphism $u:U\to X$ is an \emph{$\mathcal{M}, \mathcal{N}$-union} if there exist $m\in \mathcal{M}$ and $n\in \mathcal{M}\cap \mathcal{N}$ such that, in the poset $(\sub(X), \leq)$,
	\[[u]=[m]\vee [n]\]
	We will say that $\mathcal{M}$ is \emph{closed under $\mathcal{M},\mathcal{N}$-unions}, if it contains all such monos.  
\end{defi}

We will need a technical lemma involving kernel pairs. In order to state it, we need some preliminary observations. 

Take a category $\X$ with pullbacks endowed with a preadhesive structure $(\mathcal{M}, \mathcal{N})$, and take also the following $\mathcal{M}, \mathcal{N}$-pushout square
\[\xymatrix@R=16pt{X \ar[r]^n \ar[d]_{m}& Z\ar[d]^{p} \\ Y\ar[r]_{q} & W}\]
Pulling back $n$ and $q$ along themselves, we get two diagrams
	\[\xymatrix@R=16pt{ X\ar@/^.3cm/[drr]^{\id{X}}  \ar@/_.3cm/[ddr]_{\id{X}}\ar@{.>}[dr]^{\gamma_n}& & & Y \ar@/^.3cm/[drr]^{\id{Y}} \ar@/_.3cm/[ddr]_{\id{Y}} \ar@{.>}[dr]^{\gamma_q}\\&K_n \ar[r]^{x_1} \ar[d]_{x_2}& X \ar[d]^{n} & &K_q \ar[r]^{y_1} \ar[d]_{y_2}& Y \ar[d]^{q}\\  &X \ar[r]_{n} & Z & &Y \ar[r]_{q}& W  }\]
	with the dotted arrows $\gamma_n:Z\to K_n$ and $\gamma_q:Y\to K_q$. Moreover, we have
\begin{align*}
	q\circ m\circ x_1&=p\circ n\circ x_1=p\circ n\circ x_2=q\circ m\circ x_2
\end{align*}
thus we have an arrow $k:K_n\to K_q$ as in the following squares.
\[\xymatrix@R=16pt{K_n\ar@{.>}[d]_{k}  \ar[r]^{x_1}& X \ar[d]^{m} &  K_n\ar@{.>}[d]_{k}  \ar[r]^{x_2}& X \ar[d]^{m}\\ K_q \ar[r]_{y_1} & Y & K_q \ar[r]_{y_2} & Y}\]

We can also construct another commutative square. From the following chains of equalities
\begin{gather*}
	y_1\circ \gamma_q \circ m = \id{Y}\circ m=m=m\circ \id{X}=m\circ x_1\circ \gamma_n=y_1 \circ k \circ \gamma_n\\
	y_2\circ \gamma_q \circ m = \id{Y}\circ m=m=m\circ \id{X}=m\circ x_2\circ \gamma_n=y_2 \circ k \circ \gamma_n
\end{gather*}
we can deduce that $ k\circ \gamma_n = \gamma_{q}\circ m$.
 
\begin{lem}[Cfr.~\cite{garner2012axioms}, Lemma 2.5]\label{lem:fon}
Let $(\mathcal{M}, \mathcal{N})$ be a preadhesive structure on a category $\X$ with pullbacks such that  $\mathcal{M}\subseteq \ad{N}$, $\mathcal{M}\cap\mathcal{N}$ contains every split mono  and  $\mathcal{M}$ is closed under $\mathcal{M},\mathcal{N}$-unions. Then given  an $\mathcal{M}, \mathcal{N}$-pushout square 
\[\xymatrix@R=16pt{X \ar[r]^n \ar[d]_{m}& Z\ar[d]^{p} \\ Y\ar[r]_{q} & W}\]
all the squares in the following diagrams, constructed as above, are stable pushouts and pullbacks.
\[\xymatrix@R=16pt{X \ar[r]^-{\gamma_n} \ar[d]_{m} & K_n \ar[d]_{k} \ar[r]^{x_1}& X \ar[r]^n \ar[d]_{m}& Z\ar[d]^{p} & X \ar[r]^-{\gamma_n} \ar[d]_{m}& K_n \ar[d]_{k} \ar[r]^{x_2}& X \ar[r]^n \ar[d]_{m}& Z\ar[d]^{p}\\Y \ar[r]_-{\gamma_q}& K_q \ar[r]_{y_1}& Y\ar[r]_{q} & W & Y \ar[r]_-{\gamma_q} & K_q \ar[r]_{y_2}& Y\ar[r]_{q} & W}\]
\end{lem}
\begin{proof}
	The rightmost square in both diagrams is a pushout by hypothesis, since it is an $\mathcal{M}, \mathcal{N}$-pushout and $m$ is $\mathcal{N}$-adhesive.	Now, by Lemma \ref{lem:pb1} the rectangles
	\[\xymatrix@R=16pt{ K_n \ar[d]_{x_1} \ar[r]^{x_2}& X \ar[r]^{m} \ar[d]_{n}& Y\ar[d]^{q}  & K_n \ar[d]_{x_2} \ar[r]^{x_1}& X \ar[r]^{m} \ar[d]_{n}& Y\ar[d]^{q}\\X \ar[r]_{n}& Y\ar[r]_{p} & W &  X \ar[r]_{n}& Y\ar[r]_{p} & W}\]	
	are pullbacks, but then also the following rectangles are pullbacks.
	\[\xymatrix@R=16pt{ K_n \ar@/^.4cm/[rr]^{m\circ x_2}\ar[d]_{x_1} \ar[r]_{k}& K_q \ar[r]_{y_2} \ar[d]_{y_1}& Y\ar[d]^{q}  & K_n \ar[d]_{x_2} \ar@/^.4cm/[rr]^{m\circ x_1} \ar[r]_{k}& K_q \ar[r]_{y_1} \ar[d]_{y_2}& Y\ar[d]^{q}\\X \ar@/_.4cm/[rr]_{p\circ n}\ar[r]^{m}& Y\ar[r]^{q} & W &  X \ar@/_.4cm/[rr]_{p\circ n} \ar[r]^{m}& Y\ar[r]^{q} & W}\]
	Therefore their left halves, which are the  central squares of the original diagrams, are pullbacks too. In particular this shows that $k$ belongs to $\mathcal{M}$ and thus it is $\mathcal{N}$-adhesive. We can now consider the following two cubes in which all faces are pullbacks
	\[\xymatrix@C=13pt@R=13pt{&K_n\ar[dd]|\hole_(.65){x_2}\ar[rr]^{x_1} \ar[dl]_{k} && X \ar[dd]^{n} \ar[dl]_{m} & & K_n\ar[dd]|\hole_(.65){x_1}\ar[rr]^{x_2} \ar[dl]_{k} && X \ar[dd]^{n} \ar[dl]_{m}\\ K_q  \ar[dd]_{y_2}\ar[rr]^(.65){y_1} & & Y \ar[dd]_(.3){q}&& K_q  \ar[dd]_{y_1}\ar[rr]^(.65){y_2} & & Y \ar[dd]_(.3){q}\\&X\ar[rr]|\hole^(.65){n} \ar[dl]_{m} && Z \ar[dl]^{p} && X\ar[rr]|\hole^(.65){n} \ar[dl]_{m} && Z \ar[dl]^{p}\\Y \ar[rr]_{q} & & W && Y \ar[rr]_{q} & & W}\]
	which prove that the two central squares in the original diagram are also pushouts.
	
	We are left with the last square. We can deduce that is a pullback applying Lemma \ref{lem:pb1} to the rectangle
	\[\xymatrix@R=16pt{X \ar@/^.4cm/[rr]^{\id{X}}\ar[r]_{\gamma_n} \ar[d]_{m} & K_n \ar[d]_{k} \ar[r]_{x_1}& X  \ar[d]^{m}\\Y \ar@/_.4cm/[rr]_{\id{X}}\ar[r]^{\gamma_q}& K_q \ar[r]^{y_1}& Y}\]
	
	 By construction $\gamma_n$ is a split mono thus it is in $\mathcal{N}$. By hypothesis $m\in\mathcal{M}$ is $\mathcal{N}$-adhesive and we can build the following diagram, in which the inner square is a pushout.
	\[\xymatrix@R=16pt{X\ar[r]^-{\gamma_n} \ar[d]_{m} & K_n \ar[d]^{p_1} \ar@/^.3cm/[ddr]^{k}\\
	Y\ar[r]_{p_2} \ar@/_.3cm/[drr]_-{\gamma_q} & E\ar@{.>}[dr]^{e} \\ && K_q}\]
We already know that the outer edges form a pullback square. $\gamma_q$ is in $\mathcal{N}$ because it is a split mono and $k$ is $\mathcal{N}$-adhesive, thus  by Proposition \ref{prop:uni} we get  a mono $e:E\to K_q$ filling the diagram and such that
$[e]=[k]\vee [\gamma_q]$.
Since $\gamma_q$ is also in $\mathcal{M}$, $e$ is an $\mathcal{M},\mathcal{N}$-union and thus it belongs to $\mathcal{M}$. Now, by construction
\begin{align*}
	\id{Y}\circ m &=m  = m\circ \id{X}=m\circ x_1\circ \gamma_n
\end{align*}
thus there exists a $h:E \to Y$ filling the diagram
\[\xymatrix@C=25pt@R=16pt{X \ar@/^.4cm/[rr]^{\id{X}}\ar[d]_{m} \ar[r]_{\gamma_n}& K_n \ar[d]^{p_1} \ar[r]_{x_1} &X \ar[d]^{m}\\ Y \ar@/_.4cm/[rr]_{\id{Y}}\ar[r]^{p_2} & E\ar@{.>}[r]^{h} & Y}\]
In this diagram the left square and the whole rectangle are pushouts, thus by Lemma \ref{lem:po} the right square is a pushout too. Now, $x_1\in \mathcal{N}$ as it is the pullback of $n$, and thus $h$ belongs to $\mathcal{N}$ too.  On the other hand we have already proved that in diagram
\[\xymatrix@R=16pt{ K_1 \ar@/^.4cm/[rr]^{k} \ar[d]_{x_1} \ar[r]_{p_1}& E\ar[d]^{h} \ar[r]_{e}& K_q\ar[d]^{y_1}\\X \ar@/_.4cm/[rr]_{m}\ar[r]^m & Y \ar[r]^{\id{Y}}& Y}\]
the whole rectangle is a pushout hence, using again Lemma \ref{lem:po}, it follows that its right half
is a pushout too. By hypothesis $e$ is $\mathcal{N}$-adhesive and thus the previous square is also a pullback, showing that $e$ is an isomorphism. 

We are left with stability: $n\in \mathcal{N}$ by hypothesis, $\gamma_n$ is in $\mathcal{N}$ because it is a split mono and $x_1$ and $x_2$ belongs to $\mathcal{N}$ as they are pullbacks of $n$. Since we have proved that $m$ and $k$ are in $\mathcal{M}$ we know that they are $\mathcal{N}$-adhesive and we can conclude.
\end{proof}

Our next step is proving that if $\mathcal{M}$ contains only $\mathcal{N}$-adhesive morphisms then $\mathcal{M}, \mathcal{N}$-pushouts are almost Van Kampen squares.

\begin{lem}\label{lem:q3/4} Let $\X$ be a category with pullbacks and consider the following cube in which the left, back, bottom and top faces are pullbacks.
\[\xymatrix@C=13pt@R=13pt{&X'\ar[dd]|\hole_(.65){x}\ar[rr]^{n'} \ar[dl]_{m'} && Z' \ar[dd]^{z} \ar[dl]_{p'} \\ Y'  \ar[dd]_{y}\ar[rr]^(.65){q'} & & W' \ar[dd]_(.3){w}\\&X\ar[rr]|\hole^(.65){n} \ar[dl]^{m} && Z \ar[dl]^{p} \\Y \ar[rr]_{q} & & W}\]
Suppose that $p$ and $p'$ are monos and that the top face is a stable pushout. Then the right face is a pullback.
\end{lem}
\begin{proof}Since $p$ is a mono, by Lemma \ref{lem:pb1}, the rectangle
	\[\xymatrix@R=16pt{Z' \ar[d]_{\id{Z'}} \ar[r]^{z}& Z \ar[d]_{\id{Z}}\ar[r]^{\id{Z}} & Z \ar[d]^{p}\\ Z' \ar[r]_z & Z \ar[r]_{p}& W}\]
	is a pullback. Take now the following three diagrams
	\[\xymatrix@R=16pt{X' \ar[r]^{n'} \ar[d]_{m'}& Z'\ar[d]^{p'} &X' \ar@/^.5cm/[rr]^{n\circ x}\ar[r]_{n'} \ar[d]_{m'}& Z'\ar[d]_{p'}  \ar[r]_{z}& Z \ar[d]^{p} &Z' \ar@/^.5cm/[rr]^{z}\ar[r]_{\id{Z'}} \ar[d]_{\id{Z'}}& Z'\ar[d]_{p'}  \ar[r]_{z}& Z \ar[d]^{p}\\ Y' \ar[r]_{q'} & W' & Z'\ar[r]^{q'} \ar@/_.5cm/[rr]_{q\circ y} & W'  \ar[r]^{w}& W& Z'\ar[r]^{p'} \ar@/_.5cm/[rr]_{p\circ z} & W'  \ar[r]^{w}& W}\]
	By hypothesis the first square is a stable pushout and the left half of the first rectangle is a pullback. Since also the bottom face is a pullback by hypothesis, it follows that the whole first rectangle is a pullback too. By the previous observation, the whole second rectangle is a pullback and, since $p'$ is a mono, its first half is a pullback square. We can then apply Lemma \ref{lem:pb2} to get the thesis.
\end{proof}

\begin{cor}\label{lem:3/4}
Let $(\mathcal{M}, \mathcal{N})$ be a preadhesive structure on a category $\X$ with pullbacks and suppose that every arrow in $\mathcal{M}$ is $\mathcal{N}$-adhesive. For every $m\in \mathcal{M}$, $n\in \mathcal{N}$ and cube
\[\xymatrix@C=13pt@R=13pt{&X'\ar[dd]|\hole_(.65){x}\ar[rr]^{n'} \ar[dl]_{m'} && Z' \ar[dd]^{z} \ar[dl]_{p'} \\ Y'  \ar[dd]_{y}\ar[rr]^(.65){q'} & & W' \ar[dd]_(.3){w}\\&X\ar[rr]|\hole^(.65){n} \ar[dl]^{m} && Z \ar[dl]^{p} \\Y \ar[rr]_{q} & & W}\]
if the top and bottom faces are pushouts and the left and back ones are pullbacks then the right face is a pullback.
\end{cor}

We are now ready to prove our theorem, generalizing \cite[Prop.~3.1]{garner2012axioms}.

\begin{thm}\label{thm:primo}Let $(\mathcal{M}, \mathcal{N})$ be a preadhesive structure on a category $\X$ with pullbacks and suppose that every split mono is in $\mathcal{M}\cap\mathcal{N}$, $\mathcal{M}\subseteq \ad{N}$ and $\mathcal{M}$ is closed under $\mathcal{M},\mathcal{N}$-unions, then $\X$ is $\mathcal{M}, \mathcal{N}$-adhesive. 
\end{thm}
\begin{proof}
Every element of $\mathcal{M}$ is $\mathcal{N}$-adhesive, thus we already know that $\mathcal{M},\mathcal{N}$-pushouts exist and are stable.
Since $\X$ has all pullbacks by hypothesis, all that we have to show  is the  remaining half of the Van Kampen condition.  Take a cube in which $m\in\mathcal{M}$, $n\in\mathcal{N}$ and such that the top and bottom faces are pushouts and the left and back ones are pullbacks
\[\xymatrix@C=13pt@R=13pt{&X'\ar[dd]|\hole_(.65){x}\ar[rr]^{n'} \ar[dl]_{m'} && Z' \ar[dd]^{z} \ar[dl]_{p'} \\ Y'  \ar[dd]_{y}\ar[rr]^(.65){q'} & & W' \ar[dd]_(.3){w}\\&X\ar[rr]|\hole^(.65){n} \ar[dl]^{m} && Z \ar[dl]^{p} \\Y \ar[rr]_{q} & & W}\]
$m'$ and $n'$ belong to, respectively, $\mathcal{M}$ and $\mathcal{N}$ thus the top face is a stable pushout square, which is also a pullback. By Corollary \ref{lem:3/4} we already know that the right face is a pullback, let us prove that the other one is a pullback too.

By Lemma \ref{lem:fon}, in the following diagrams all squares are stable pushouts and pullbacks. 
\[\xymatrix@R=16pt{X \ar[r]^-{\gamma_n} \ar[d]_{m} & K_n \ar[d]_{k} \ar[r]^{x_1}& X \ar[r]^n \ar[d]_{m}& Z\ar[d]^{p} & X \ar[r]^-{\gamma_n} \ar[d]_{m}& K_n \ar[d]_{k} \ar[r]^{x_2}& X \ar[r]^n \ar[d]_{m}& Z\ar[d]^{p}\\Y \ar[r]_-{\gamma_q}& K_q \ar[r]_{y_1}& Y\ar[r]_{q} & W & Y \ar[r]_-{\gamma_q} & K_q \ar[r]_{y_2}& Y\ar[r]_{q} & W\\X' \ar[r]^-{\gamma_{n'}} \ar[d]_{m'} & K_{n'} \ar[d]_{k'} \ar[r]^{x'_1}& X' \ar[r]^{n'} \ar[d]_{m'}& Z'\ar[d]^{p'} & X' \ar[r]^-{\gamma_{n'}} \ar[d]_{m'}& K_{n'} \ar[d]_{k'} \ar[r]^{x'_2}& X' \ar[r]^{n'} \ar[d]_{m'}& Z'\ar[d]^{p'}\\Y' \ar[r]_-{\gamma_{q'}}& K_{q'} \ar[r]_{y'_1}& Y\ar[r]_{q'} & W' & Y \ar[r]_-{\gamma_{q'}} & K_{q'} \ar[r]_{y'_2}& Y\ar[r]_{q'} & W}\]

By Corollary \ref{cor:cube}, there exist the dotted $t_1:K_{n'}\to K_n$ and $t_2:K_{q'}\to K_q$ as below: 
	\[\xymatrix@C=13pt@R=13pt{&K_{n'}\ar[dd]|\hole_(.65){x'_2}\ar[rr]^{x'_1} \ar@{.>}[dl]_{t_1} && X' \ar[dd]^{n'} \ar[dl]_{x} & & K_{q'}\ar[dd]|\hole_(.65){y'_2}\ar[rr]^{y'_1} \ar@{.>}[dl]_{t_2} && Y' \ar[dd]^{q'} \ar[dl]_{y}\\ K_n  \ar[dd]_{x_2}\ar[rr]^(.65){x_1} & & X\ar[dd]_(.3){n}&& K_q  \ar[dd]_{y_2}\ar[rr]^(.65){y_1} & & Y \ar[dd]_(.3){q}\\&X'\ar[rr]|\hole^(.65){n'} \ar[dl]_{x} && Z' \ar[dl]^{z} && Y'\ar[rr]|\hole^(.65){q'} \ar[dl]_{y} && W' \ar[dl]^{w}\\X \ar[rr]_{n} & & Z && Y \ar[rr]_{q} & & W}\]
and the left face of the first cube is a pullback square.  If we compute we get
\begin{align*}
	x_1\circ t_1\circ \gamma_{n'}& =x\circ x'_1\circ \gamma_{n'}=x\circ \id{X'} =\id{X} \circ x=x_1\circ \gamma_n \circ x
\\	
	y_1\circ t_2\circ \gamma_{q'}& =y\circ y'_1\circ \gamma_{q'}=y\circ \id{Y'} =\id{Y} \circ y=y_1\circ \gamma_q \circ y
\\
	y_1\circ t_2\circ k' & =y\circ y'_1\circ k'=y\circ m'\circ x'_1 = m\circ x\circ x'_1=m\circ x_1\circ t_1=y_1\circ k\circ t_1
\\
	x_2\circ t_1\circ \gamma_{n'}&=x\circ x'_2\circ \gamma_{n'} =x\circ \id{X'}
	=\id{X} \circ x=x_2\circ \gamma_n \circ x
\\
	y_2\circ t_2\circ \gamma_{q'}&=y\circ y'_2\circ \gamma_{q'}=y\circ \id{Y'}=\id{Y} \circ y=y_2\circ \gamma_q \circ y 
\\
	y_2\circ t_2\circ k'&=y\circ y'_2\circ k'=y\circ m'\circ x'_2 =m\circ x\circ x'_2=m\circ x_2\circ t_1=y_2\circ k\circ t_1
\end{align*}

therefore the following three squares commute
\[\xymatrix@R=16pt{X'  \ar[r]^{\gamma_{n'}} \ar[d]_{x}& K_{n'} \ar[d]^{t_1}& Y' \ar[d]_{y} \ar[r]^{\gamma_{q'}} & K_{q'}\ar[d]^{t_2}& K_{n'}\ar[r]^{k'} \ar[d]_{t_1} & K_{q'}\ar[d]^{t_2}\\ X \ar[r]_{\gamma_n}& K_n  & Y \ar[r]_{\gamma_q} & K_q& K_n\ar[r]_{k} & K_q}
\]
The first one of the squares above is a pullback: this follows applying Lemma \ref{lem:pb1} to the rectangle below.
\[\xymatrix@R=16pt{X' \ar[r]_{\gamma_{n'}} \ar[d]_{x}\ar@/^.4cm/[rr]^{\id{X'}}& K_{n'} \ar[r]_{x'_2} \ar[d]_{t_1} & X' \ar[d]^{x} \\ X \ar[r]^{\gamma_n} \ar@/_.4cm/[rr]_{\id{X}}& K_{n} \ar[r]^{x_2}& X}\]
We can then use these arrows $t_1$ and $t_2$ to construct the following cube
		\[\xymatrix@C=13pt@R=13pt{&X'\ar[dd]|\hole_(.65){x}\ar[rr]^{m'} \ar[dl]_{\gamma_{n'}} && Y' \ar[dd]^{y} \ar[dl]_{\gamma_{q'}}\\ K_{n'}  \ar[dd]_{t_1}\ar[rr]^(.65){k'} & & K_{q'}\ar[dd]_(.3){t_2}\\&X\ar[rr]|\hole^(.65){m} \ar[dl]_{\gamma_n} && Y \ar[dl]^{\gamma_q} \\K_n \ar[rr]_{k} & & K_q}\]
which has pullbacks as left and back faces and stable pushouts as top and bottom ones. $\gamma_q$ and $\gamma_{q'}$ are split monos, thus by Lemma \ref{lem:q3/4} the right face  is a pullback. Switching $\gamma_n$ and $m$ we get another cube		\[\xymatrix@C=13pt@R=13pt{&X'\ar[dd]|\hole_(.65){x}\ar[rr]^{\gamma_{n'}} \ar[dl]_{m'} && K_{n'} \ar[dd]^{t_1} \ar[dl]_{k'}\\ Y'  \ar[dd]_{y}\ar[rr]^(.65){\gamma_{q'}} & & K_{q'}\ar[dd]_(.3){t_2}\\&X\ar[rr]|\hole^(.65){\gamma_n} \ar[dl]_{m} && K_n \ar[dl]^{k} \\Y \ar[rr]_{\gamma_q} & & K_q}\]
to which we can apply  Corollary \ref{lem:3/4}, to get again that the right face is a pullback. Now, by Lemma \ref{lem:pb1} the following rectangle is a pullback 
\[\xymatrix@R=16pt{K_{n'} \ar[r]^{x'_1} \ar[d]_{t_1}& X' \ar[d]_{x} \ar[r]^{m'} & Y'\ar[d]^{y}\\K_n\ar[r]_{x_1} & X \ar[r]_{m} & Y}\]
Thus we can apply Lemma \ref{lem:pb2} to the diagrams
\[\xymatrix@R=16pt{X \ar[r]^{\gamma_{n}} \ar[d]_{m}& K_n\ar[d]^{k} &Y' \ar@/^.4cm/[rr]^{\id{Y'}}\ar[r]_{\gamma_{q'}} \ar[d]_{y}& K_{q'}\ar[d]_{t_2}  \ar[r]_{y'_2}& Y' \ar[d]^{y} &K_{n'} \ar@/^.4cm/[rr]^{m'\circ x'_1}\ar[r]_{k'} \ar[d]_{t_1}& K_{q'}\ar[d]_{q'}  \ar[r]_{y'_2}& Y' \ar[d]^{y}\\ Y \ar[r]_{\gamma_q} & K_q & Y\ar[r]^{\gamma_q} \ar@/_.4cm/[rr]_{\id{Y}} & K_q  \ar[r]^{y_2}& Y& K_n \ar@/_.4cm/[rr]_{m\circ x_1}\ar[r]^{k} & K_q  \ar[r]^{y_2}& Y}\]
to deduce that the square below is a pullback too.
\[\xymatrix@R=16pt{K_{q'}\ar[d]_{t_2}  \ar[r]^{y'_2}& Y' \ar[d]^{y} \\ K_q  \ar[r]_{y_2}& Y}\]	
	Which, in turn also entails that the following rectangle is a pullback.
\[\xymatrix@R=16pt{K_{q'} \ar[r]^{t_2}  \ar[d]_{y'_2}& K_q \ar[d]_{y_2} \ar[r]^{y_1}& Y\ar[d]^{q} \\ Y' \ar[r]_{y} & Y \ar[r]_q & W}\]	
	
	We can now notice that the diagrams
	 	\[\xymatrix@R=16pt{X' \ar[r]^{n'} \ar[d]_{m'}& Z'\ar[d]^{p'} &K_{q'} \ar@/^.4cm/[rr]^{y_1\circ t_2}\ar[r]_{y'_1} \ar[d]_{y'_2}& Y'\ar[d]_{q'}  \ar[r]_{y}& Y \ar[d]^{q} &X' \ar@/^.4cm/[rr]^{m\circ x}\ar[r]_{m'} \ar[d]_{n'}& Y'\ar[d]_{q'}  \ar[r]_{y}& Y \ar[d]^{q}\\ Y' \ar[r]_{q'} & W' & Y'\ar[r]^{q'} \ar@/_.4cm/[rr]_{q\circ y} & W'  \ar[r]^{w}& W& Z' \ar@/_.4cm/[rr]_{p\circ z}\ar[r]^{p'} & W'  \ar[r]^{w}& W}\]
satisfy the hypothesis of Lemma \ref{lem:pb2} and this yields the thesis.
\end{proof}

The previous theorem yields at once the following two corollaries (compare with \cite[Thm.~A and B]{garner2012axioms}).

\begin{cor}\label{cor:qad}
Let $\X$ be a category with pullbacks, then
\begin{enumerate}
	\item if $\mon(\X)\subseteq \mor(\X)_\mathsf{a}$ then $\X$ is adhesive;
	\item if $\reg(\X)\subseteq \mor(\X)_\mathsf{a}$ and it is closed under binary joins then $\X$ is quasiadhesive.
\end{enumerate}
\end{cor}

\begin{cor}\label{cor:mads}
Let $\mathcal{M}$ be a stable system of monos in a category $\X$ with pullbacks. Suppose that $\mathcal{M}$ is stable under pushouts, it contains all split monos,  it is closed under binary joins  and $\mathcal{M}\subseteq \mor(\X)_{\mathsf{a}}$, then $\X$ is an $\mathcal{M}$-adhesive category.
\end{cor}

\begin{rem} In Corollaries \ref{cor:qad} and \ref{cor:mads},  closure under joins means that, given  $m:M\to X$, $n:N\to X$ in $\reg(\X)$ or in $\mathcal{M}$, any representative of $[m]\vee[n]$, which exists by virtue of Proposition \ref{prop:uni}, is again in $\reg(\X)$ or in $\mathcal{M}$.
\end{rem}

\paragraph{Application to toposes}
In \cite{lack2006toposes} it is proved that any elementary topos is adhesive, using descent techniques. In this section we will provide an alternative proof, already hinted at in \cite{garner2012axioms}, based on what we have shown in this section. Our main references for topos theory are \cite{johnstone2002sketches1,johnstone2002sketches,maclane2012sheaves,mclarty1992elementary}.

\begin{defi} Let $\X$ be a finitely complete category. A \emph{subobject classifier} is a mono $\true:1\to \Omega$ such that, for every monomorphism $m:M\to X$, there is a unique $\chi_m:X\to \Omega$ such that the square below is a pullback
	\[\xymatrix@R=16pt{M \ar[d]_{m}\ar[r]^{!_M} & 1 \ar[d]^{\true}\\ X \ar@{.>}[r]_{\chi_m} & \Omega}\]
	A \emph{topos}  is a finitely complete, cartesian closed category $\X$ which admits a subobject classifier.
\end{defi}

Let us state some well known properties of toposes.

\begin{prop}\label{prop:stab}
	If $\X$ is a topos, then $\X$ is finitely cocomplete and all pushout squares are stable.
\end{prop}

\begin{lemC}[{\cite[Cor.~A2.4.3]{johnstone2002sketches1}}]\label{lem:popb}
	Let $m:X\to Y$ and $f: X\to Z$ be arrows in a topos $\X$ and suppose that $m$ is a monomorphism, if the square 
	\[\xymatrix@R=16pt{X \ar[r]^{f} \ar[d]_{m}& Z \ar[d]^{q_1}\\ Y \ar[r]_{q_2} & Q}\]
	is a pushout, then $q_1$ is a mono and the square is also a pullback.
\end{lemC}

From Proposition \ref{prop:stab} and Lemma \ref{lem:popb} we can easily deduce the following.
\begin{cor}
	In a topos $\X$, every mono is adhesive.
\end{cor}

We can now apply Corollary \ref{cor:mequiv}, Lemma \ref{lem:popb} and Remark \ref{rem:ex} to get our result.

\begin{cor}\label{cor:topad}
	Every topos is an adhesive category.
\end{cor}

\section{From $\mathcal{M},\mathcal{N}$-adhesivity to $\mathcal{M},\mathcal{N}$-unions}\label{sec:aduni}

In the previous section we deduced $\mathcal{M},\mathcal{N}$-adhesivity from the closure of $\mathcal{M}$ under some kinds of unions. In this section we will go in the opposite direction.

\begin{defi}

Let $f:X\to Y$ be an arrow in a category $\X$ such that pushout square below  exists.
\[\xymatrix@R=16pt{ X \ar[r]^f \ar[d]_f& Y \ar[d]^{y_1}\\ Y \ar[r]_{y_2} & Q_f}\]
The \emph{codiagonal} $\upsilon_f:Q_f\to Y$ is the unique arrow  fitting in the following diagram.
\[\xymatrix@R=16pt{ X \ar[r]^f \ar[d]_f& Y \ar@/^.3cm/[ddr]^{\id{Y}} \ar[d]^{y_1}\\ Y \ar@/_.3cm/[drr]_{\id{Y}}\ar[r]_{y_2} & Q_f \ar@{.>}[dr]^{\upsilon_f}\\ && Y}\]

Given a preadhesive structure $(\mathcal{M}, \mathcal{N})$, a \emph{$\mathcal{M}, \mathcal{N}$-codiagonal} is the codiagonal of an arrow $n\in \mathcal{M}\cap \mathcal{N}$.
\end{defi}

Let us list some useful properties of codiagonals.
\begin{lem}\label{rem:coeq}
	Let $f:X\to Y$ be a morphism in a category $\X$ and suppose that $f$ admits a codiagonal $\upsilon_f:Q_f\to Y$, then the following hold true:
	\begin{enumerate}
		\item $\upsilon_f$ is the coequalizer of the pair of coprojections $y_1, y_2:Y\rightrightarrows Q_f$;
		\item if a pullback of $y_1$ along $y_2$ exists, then the pair $y_1, y_2:Y\rightrightarrows Q_f$ has an equalizer $e:E\to Y$ and, moreover, the following square is a pullback
		\[\xymatrix@R=16pt{E \ar[r]^{e} \ar[d]_{e} & Y \ar[d]^{y_1}\\  Y\ar[r]_{y_2} &  Q_f}\]
	\end{enumerate}
\end{lem}

The first ingredient we need is a generalization of \cite[Prop. $4.4$]{garner2012axioms}.
\begin{lem}\label{lem:fact}
Let $(\mathcal{M}, \mathcal{N})$ be a preadhesive structure on a category $\X$ with pullbacks and $u:U\to X$ an $\mathcal{M}, \mathcal{N}$-union. Suppose that $\mathcal{M}\subseteq\ad{N}$, that $\mathcal{M}\cap \mathcal{N}$ contains all split monomorphisms and
that $\mathcal{N}$ contains all $\mathcal{M}, \mathcal{N}$-codiagonals, then:
\begin{enumerate}
	\item $u$ admits pushouts  along itself (i.e. it has a \emph{cokernel pair}); 
	\item there exists an epi $e_u:U\to E_u$ and an element $m_u: E_u\to X$ of $\mathcal{M}\cap \mathcal{N}$ such that   $u=m_u\circ e_u$.
\end{enumerate}
\end{lem}
\begin{rem}
Notice that, if $\mathcal{M}\subseteq \ad{N}$, then for every $n\in \mathcal{M}\cap \mathcal{N}$ a pushout square 
\[\xymatrix@R=16pt{N \ar[r]^n \ar[d]_n& X \ar[d]^{n_1} \\ X \ar[r]_{n_2} &Q_n}\]
of $n$ along itself exists, and thus there also exists the codiagonal $\upsilon_n$. 
\end{rem}
\begin{proof}[Proof of Lemma \ref{lem:fact}.]
\begin{enumerate}[leftmargin=0pt,itemindent=1.7em]
	\item 
	Let $m:M\to X$ in $\mathcal{M}$ and $n:N\to X$ in $\mathcal{M}\cap \mathcal{N}$ be such that
	\[[u]=[m]\vee [n]\] 
	By Proposition \ref{prop:uni} we can consider the following diagram,	in which the outer edges form a pullback and the inner square is a pushout.
	\[\xymatrix@R=16pt{P \ar[r]^{p_1} \ar[d]_{p_2} & M\ar[d]^{u_2} \ar@/^.3cm/[ddr]^{m} \\ N \ar@/_.3cm/[drr]_{n}\ar[r]_{u_1} & U \ar[dr]^{u} \\ && X}\]
	Pulling back $m$ along $\upsilon_n$ we get a pullback
	square
	\[\xymatrix@R=16pt{T\ar[r]^{t_1}  \ar[d]_{t_2}& M\ar[d]^{m}\\ Q_{n} \ar[r]_{\upsilon_n} & X}\]
	There exist $l_1, l_2:M\rightrightarrows T$ as in the following diagram
	\[\xymatrix@R=16pt{M\ar@{.>}[dr]_{l_1} \ar@/^.3cm/[drr]^{\id{M}} \ar[d]_{m} & & & M\ar@{.>}[dr]_{l_2} \ar@/^.3cm/[drr]^{\id{M}} \ar[d]_{m}\\X \ar@/_.3cm/[dr]_{n_1}&T\ar[r]^{t_1}  \ar[d]_{t_2}& M\ar[d]^{m} & X \ar@/_.3cm/[dr]_{n_2}&T\ar[r]^{t_1}  \ar[d]_{t_2}& M\ar[d]^{m}\\ &Q_{n} \ar[r]_{\upsilon_n} & X & &Q_{n} \ar[r]_{\upsilon_n} & X}\]
	
	By Lemma \ref{lem:pb1} the following are pullback squares
	\[\xymatrix@R=16pt{M \ar[r]^{l_1} \ar[d]_{m}& T \ar[d]^{t_2}&  M \ar[r]^{l_2} \ar[d]_{m}& T \ar[d]^{t_2}\\ X \ar[r]_{n_1} & Q_{n} & X \ar[r]_{n_2} & Q_{n}}\]
	therefore, since $n$ is $\mathcal{N}$-adhesive, the top face of the following cube is a pushout.
	\[\xymatrix@C=13pt@R=13pt{&P\ar[dd]|\hole_(.65){p_2}\ar[rr]^{p_1} \ar[dl]_{p_1} && M \ar[dd]^{m} \ar[dl]_{l_1}\\ M  \ar[dd]_{m}\ar[rr]^(.65){l_2} & & T\ar[dd]_(.3){t_2}\\&N\ar[rr]|\hole^(.65){n} \ar[dl]_{n} && X \ar[dl]^{n_1} \\X \ar[rr]_{n_2} & & Q_n}\]
	
	Now, $t_1$ is the pullback of an $\mathcal{M}, \mathcal{N}$-codiagonal, thus it is in $\mathcal{N}$, while $t_2$ is in $\ad{N}$ since it is the pullback of $m$, therefore the pushout square below exists.
	\[\xymatrix@R=16pt{T \ar[r]^{t_1} \ar[d]_{t_2}& M \ar[d]^{q_1}\\ Q_n \ar[r]_{q_2} & Q }\]
	Suppose now that the solid part of the next diagram is given
	\[\xymatrix@R=16pt{U \ar[rr]^u \ar[dd]_u && X \ar[dl]_{z_1} \ar[d]^{n_1}\\ &Z& Q_n \ar[d]^{q_2}\\ X\ar[r]_{n_2} \ar[ur]^{z_2} & Q_n \ar[r]_{q_2}  &Q \ar@{.>}[ul]_{z}}\]
	Precomposing with $u_1$ and $u_2$ we get the following identities
	\[\begin{split}
			z_1\circ m &= z_1\circ u\circ u_2 =z_2\circ u \circ u_2 =z_2\circ m 
	\end{split} \qquad \begin{split}	z_1\circ n &= z_1\circ u\circ u_1 =z_2\circ u \circ u_1 =z_2\circ n
	\end{split}\]
	The second chain of the equalities above implies the existence of the dotted $w:Q_n\to Z$.
	\[\xymatrix@R=16pt{N \ar[r]^{n} \ar[d]_{n} & X\ar[d]^{n_1} \ar@/^.3cm/[ddr]^{z_1} \\ X \ar@/_.3cm/[drr]_{z_2}\ar[r]_{n_2} & Q_n \ar@{.>}[dr]^{w} \\ && Z}\]
	If we compute we get
	\begin{align*}
		w\circ t_2\circ l_2 =w\circ n_2\circ m =z_2\circ m =z_1\circ m =w\circ n_1\circ m =w\circ t_2\circ l_1
	\end{align*}
	By construction and by our previous observations, $t_1$ is a codiagonal for $p_1$, thus Lemma \ref{rem:coeq}(1) implies the existence of a unique $k:M\to Z$ making the following diagram commutative
	\[\xymatrix@R=16pt{T \ar[r]^{t_1} \ar[d]_{t_2}& M \ar@/^.3cm/[ddr]^{k}\ar[d]^{q_1}\\ Q_n \ar@/_.3cm/[drr]_{w}\ar[r]_{q_2} & Q \ar@{.>}[dr]^{z}\\ && Z }\]
	which, in turn, implies the existence of the dotted $z$. If we compute further we have
	\[z_1=w\circ n_1=z\circ q_2\circ n_1 \qquad 
	  z_2=w\circ n_2=z\circ q_2\circ n_2
	\]
	Moreover, if $z':Q\to Z$ is such that
	\[z_1=z'\circ q_2\circ n_1 \qquad z_2=z'\circ q_2\circ n_2\]
	then we also have 
	\[
		z'\circ q_2\circ n_1 =z_1 =w\circ n_1 
	    \qquad
		z'\circ q_2\circ n_2 =z_2 =w\circ n_2
	\]
	which shows that $w=z'\circ q_2$. On the other hand
	\[z'\circ q_1\circ t_1  = z'\circ q_2\circ t_2 =w\circ t_2 \]
	and so we also have that $z'\circ q_1=k$, allowing us to conclude that $z=z'$. We can now deduce that the following square is a pushout 
	\[\xymatrix@R=16pt{U \ar[r]^u \ar[d]_u & X \ar[d]^{q_2\circ n_1}\\  X\ar[r]_{q_2\circ n_2} &  Q }\]
	\item By the previous point $u$ has pushout along itself, therefore there exists a codiagonal $\upsilon_{u}:Q\to U$. In particular, $q_2\circ n_1$ and $q_2\circ n_2$ are split monos and thus elements of $\mathcal{M}\cap\mathcal{N}$. By the second point of Lemma \ref{rem:coeq} they have an equalizer $m_u:E_u\to X$ which, since $\mathcal{M}$ and $\mathcal{N}$ are stable under pullback, is also an element of $\mathcal{M}\cap \mathcal{N}$. Since, by construction
	\[q_2\circ n_1\circ u= q_2\circ n_2\circ u\]
	we also get an arrow $e_u:U\to E_u$ such that $u=m_u\circ e_u$. To show that this arrow is epi, let us start with the equalities
	\[
		m=u\circ u_2=m_u\circ e_u\circ u_2
	 \qquad 
	n=u\circ u_2=m_u\circ e_u\circ u_1
\]
Since $\mathcal{M}$ and $\mathcal{N}$ are closed under decomposition and $\mathcal{M}$-decomposition we can deduce that $e_u\circ u_2$ belongs to $\mathcal{M}$ and that $e_u\circ u_1$ is an element of $\mathcal{M}\cap \mathcal{N}$.  

Let now $b:B\to E_u$ be another mono such that 
\[b\circ b_1 =e_u\circ u_1 \qquad b\circ b_2 =e_u\circ u_2\] 
for some $b_1:N\to B$ and $b_2:M\to B$. Then 
\[
b \circ b_1 \circ p_2  = e_u \circ u_1\circ p_2= e_u \circ  u_2\circ p_1=b\circ b_2\circ p_1\]
which, since $b$ is mono entails
\[b_1\circ  p_2 =  b_2\circ p_1\]
and thus there exists $\hat{b}:U\to B$ such that
\[b_1 =\hat{b} \circ u_1 \qquad b_2=\hat{b}\circ u_2\]
Further computing we get
\[
	b\circ \hat{b} \circ u_1 =b\circ b_1 =e_u\circ u_1 
\qquad 
b\circ \hat{b} \circ u_2 =b\circ b_2 =e_u\circ u_2 
\]
which shows that $[e_u] \leq [b]$ in $(\sub(E_u), \leq)$. Thus $e_u$ is a union for $e_u\circ u_2$ and $e_u\circ u_1$. 

By the previous point and point $2$ of Lemma \ref{rem:coeq} we get the diagram below, in which the outer edges form a pushout, the inner square is a pullback and $c$ is the equalizer of $c_1$ and $c_2$.
\[\xymatrix@R=16pt{U  \ar@/^.5cm/[rr]^{e_u} \ar@/_.4cm/[dr]_{e_u}  \ar@{.>}[r]_{e}& C \ar[d]_{c} \ar[r]_{c} & E_u \ar[d]^{c_2} \\  & E_u \ar[r]_{c_1} & \hat{Q}}\]
If we show that $c$ is invertible we are done: in such a case $c_1$ must be equal to $c_2$ and this implies that $e_u$ is an epimorphism. The existence of $e:U\to C$ can then be inferred from the universal property of pullbacks. Notice, moreover, that $c_1$ and $c_2$ are in $\mathcal{M}\cap \mathcal{N} $ since they are split monos, thus $c\in \mathcal{M}\cap \mathcal{N}$ too. Suppose that the solid part of the following diagram is given.
\[\xymatrix@R=16pt{C\ar[r]^c \ar[d]_{c}& E_u \ar[r]^{m_u} & X \ar[d]^{n_1}  \ar[dl]_{z_1}\\ E_u \ar[d]_{m_u} &Z& Q_n \ar[d]^{q_2}\\ X \ar[ur]^{z_2} \ar[r]_{n_2}&Q_n \ar[r]_{q_2}& Q \ar@{.>}[ul]^{z} }\]
Then we have
\[
z_1\circ u = z_1\circ m_u\circ e_u = z_1\circ m_u\circ c \circ e = z_2\circ m_u\circ c \circ e = z_2\circ m_u \circ e_u= z_2\circ u \]
and thus there exists $z:Q\to Z$ such that 
\[z_1= z\circ q_2\circ n_1\qquad z_2=z\circ q_2\circ n_2\]
Uniqueness of such a $z$ follows at once since $q_2\circ n_1$ and $q_2\circ n_2$ are the coprojections of a pushout , thus we can conclude that the square below is a pushout.
\[\xymatrix@C=40pt@R=16pt{U \ar[r]^{m_u\circ c} \ar[d]_{m_u\circ c} & X \ar[d]^{q_2\circ n_1}\\  X\ar[r]_{q_2\circ n_2} &  Q }\]
Now, $\mathcal{M}$ and $\mathcal{N}$ are closed under composition, thus $m_u\circ c$ is in $\mathcal{M}\cap \mathcal{N}$ and, since $\mathcal{M}\subseteq \ad{N}$, it follows from the third point of Proposition \ref{prop:dec} that $m_u\circ c$ is a regular mono.  In every category a regular mono is the equalizer of its cokernel pair, thus $m_u\circ c$ is the equalizer of $q_2\circ n_1$ and $q_2\circ n_2$ and therefore $c$ must be an isomorphism.	\qedhere 
\end{enumerate}
\end{proof}

We are now ready to prove the main theorem of this section (see \cite[Thm. $19$]{johnstone2007quasitoposes}).
\begin{thm}\label{thm:secondo}Let $\X$ be an $\mathcal{M}, \mathcal{N}$-adhesive  category with pullbacks. If $\mathcal{M}\cap \mathcal{N}$ contains all split monomorphisms and $\mathcal{N}$ contains all $\mathcal{M}, \mathcal{N}$-codiagonals, then $\mathcal{M}$ is closed under $\mathcal{M}, \mathcal{N}$-unions.
\end{thm}
\begin{proof}	Let $u:U\to X$ be the $\mathcal{M}, \mathcal{N}$-union of $m:M\to X$ in $\mathcal{M}$ and $n:N\to X$ in $\mathcal{M}\cap \mathcal{N}$, by Example \ref{ex:ade} and Proposition \ref{prop:uni} we know that these arrows fit in a diagram
	\[\xymatrix@R=16pt{P \ar[r]^{p_1} \ar[d]_{p_2} & M\ar[d]^{u_2} \ar@/^.3cm/[ddr]^{m} \\ N \ar@/_.3cm/[drr]_{n}\ar[r]_{u_1} & U \ar[dr]^{u} \\ && X}\]
in which the outer edges form a pullback and the inner square is a pushout.  Notice that $p_2\in \mathcal{M}$ and $p_1\in \mathcal{N}$ thus, by Proposition \ref{prop:pbpo} the inner square is also a pullback. By Lemma \ref{lem:fact} we also know that $u=m_u\circ e_u$ for some epi $e_u:Y\to E_u$ and $m_u:E_u\to X$ in $\mathcal{M}\cap \mathcal{N}$. As we have noticed before, the decomposition properties of $\mathcal{M}$ and $\mathcal{N}$ imply that $e_u\circ u_2\in \mathcal{M}$ and $e_u\circ u_1 \in \mathcal{M}\cap \mathcal{N}$.   Our strategy to prove the theorem consists in  showing that $e_u$ is an isomorphism.

First of all notice that $e_u$ is a mono because $u=m_u\circ e_u$, thus in the following diagram every square is a pullback 
and applying Lemma \ref{lem:pb1} we can deduce that the composite square is a pullback too.
\[\xymatrix@R=16pt{P \ar[r]^{p_1} \ar[d]_{p_2}& M \ar[r]^{\id{M}} \ar[d]_{u_2} & M \ar[d]^{u_2} \\ N \ar[r]_{u_1}  \ar[d]_{\id{U}}& U \ar[r]^{\id{U}} \ar[d]^{\id{U}} & U \ar[d]^{e_u} \\ N  \ar[r]_{u_2}& U \ar[r]_{e_u} & E_u}\]

Next,  since the arrow $n$ is in $\mathcal{M}$,  $p_1$ is in $\mathcal{M}$ as it is its pullback and $u_1\in \mathcal{M}$ since it is the pushout of $p_1$. We can then build the following two pushout squares, which, by Proposition \ref{prop:pbpoad}, are also pullbacks.
\[\xymatrix@C=40pt@R=16pt{N \ar[r]^-{e_u\circ u_1} \ar[d]_{e_u\circ u_1} & E_u \ar[d]^{e_1} & N \ar[r]^-{e_u\circ u_1} \ar[d]_{ u_1} & E_u \ar[d]^{a_1} \\ E_u   \ar[r]_-{e_2} & Q_{e_u\circ u_1} & U \ar[r]_{a_2}& A}\]
which can be merged together: indeed the solid part of the following diagram is commutative, thus the dotted arrow $a$ exists. Moreover, by Lemma \ref{lem:po}, the bottom rectangle is a pushout.
\[\xymatrix@R=16pt{N \ar[r]^{u_1} \ar[d]_{u_1} & U \ar[r]^{e_u} & E_u \ar[d]_{a_1} \ar@/^.5cm/[dd]^{e_1} \\
U \ar[d]_{e_u} \ar[rr]^{a_2}&& A \ar@{.>}[d]_{a} \\ E_u \ar[rr]_{e_2}&&Q_{e_u\circ u_1}}\]
Notice that, since $u_1\in \mathcal{M}$ and $e_u\circ u_1$ is in $\mathcal{N}$, the upper half of the square above is also a pullback.

Now, $e_2$ is the pushout of $e_u\circ u_1$, thus it is in $\mathcal{M}$ and so it is a mono. This, together with Lemma \ref{lem:pb1}, entails that the following rectangle is a pullback.
\[\xymatrix@R=16pt{U \ar[r]^{e_u}  \ar[d]_{\id{U}}& E_u \ar[r]^-{\id{E_u}} \ar[d]_{\id{E_u}} & E_u \ar[d]^{e_2}\\ U \ar[r]_{e_u} & E_u \ar[r]_-{e_2} & Q_{e_u\circ u_1}}\]
The arrow $a_2$ is in $\mathcal{M}$ as it is the pushout of $e_u\circ u_1$. By applying Lemma \ref{lem:pb2} to the diagrams
\[\xymatrix@R=16pt{ N \ar[rr]^-{e_u\circ u_1} \ar[d]_{ u_1} && E_u \ar[d]^{a_1} & N \ar[d]_{e_u\circ u_1} \ar[r]^{u_1} & U \ar[d]_{a_2}\ar[r]^-{e_u} & E_u  \ar[d]^{e_2}&  U  \ar[r]^{\id{U}} \ar[d]_{\id{U}}& U \ar[d]_{a_2} \ar[r]^-{e_u} & E_u \ar[d]^{e_2} \\ U \ar[rr]_{a_2}&& A & E_u \ar[r]_{a_1} &  A \ar[r]_-{a} & Q_{e_u\circ u_1} & U \ar[r]^{a_2} \ar@/_.4cm/[rr]_{e_2\circ e_u} & A \ar[r]^-{a} & Q_{e_u\circ u_1}}\]
we get that also the following square is a pullback too.
\[\xymatrix@R=16pt{
	U \ar[d]_{e_u} \ar[r]^-{a_2}& A \ar[d]^{a} \\ E_u \ar[r]_-{e_2}&Q_{e_u\circ u_1}}\]

On the other hand, the arrow $p_1:P\to M$ is in $\mathcal{M}\cap \mathcal{N}$ as it is the pullback of $n$, thus it admits a pushout  the following pushout square along itself
\[\xymatrix@R=16pt{P \ar[r]^{p_1} \ar[d]_{p_1}& M \ar[d]^{m_1} \\ M \ar[r]_{m_2} & Q_{p_1}}\]
We can then construct the solid part of the rightmost rectangle in the diagram below, inducing the dotted $b:Q_{p_1}\to A$. Notice that the first rectangle is a pushout by Lemma \ref{lem:po} so that  the right half of the second diagram also is a pushout, again because of Lemma \ref{lem:po}, and $b$ belongs to $\mathcal{M}$.
\[\xymatrix@R=16pt{P \ar[r]^{p_2} \ar[d]_{p_1}& N \ar[d]_{u_1} \ar[r]^{u_1}& U \ar[r]^{e_u} & E_u \ar[d]^{a_1} & P \ar@/^.4cm/[rrr]^{e_u\circ u_1\circ p_2}\ar[r]_{p_1} \ar[d]_{p_1}& M \ar[d]_{m_1} \ar[r]_{u_2}& U \ar[r]_{e_u}& E_u \ar[d]^{a_1}\\ M \ar[r]_{u_2}& U \ar[rr]_{a_2}&& A& M \ar@/_.4cm/[rrr]_{a_2\circ u_2}\ar[r]^{m_2} & Q_{p_1} \ar@{.>}[rr]^{b} && A}\]
We can compose with the codiagonal $\upsilon_{p_1}:Q_{p_1}\to M$ to obtain the solid diagram
\[\xymatrix@R=16pt{M \ar@/_.5cm/[dd]_{\id{M}}\ar[r]^{u_2} \ar[d]^{m_1} & U \ar[r]^-{e_u} & E_u \ar[d]_{a_1} \ar@/^.5cm/[dd]^{\id{E_u}} \\
	Q_{p_1} \ar[d]^{\upsilon_{p_1}} \ar[rr]^{b}&& A \ar@{.>}[d]_{r} \\ M\ar[r]_{u_2} & U\ar[r]_-{e_u}&E_u}\]	
Since the upper half of the square above is a pushout then the dotted $r:A\to Q_{e_u\circ u_1}$ exists. Moreover, since the outer edges make a pushout square, the lower half is a pushout too, by Lemma \ref{lem:po}, and, because $\upsilon_{p_1}$ belongs to $\mathcal{N}$,  also a pullback, by Proposition \ref{prop:pbpoad}.

We can now notice that for every $z_1: Z\to M$ and $z_2:Z\to E_u$ such that  $m\circ z_1=m_u\circ z_2$ we have the following chain of equalities
\[
m_u\circ e_u\circ u_2\circ z_1=u\circ u_2\circ z_1=m\circ z_1=m_u\circ z_2 \]
which, since $m_u$ is mono, entails that $z_2=e_u\circ u_2\circ z_1$.

This, in turn, can be rephrased saying that the square below is a pullback 
\[\xymatrix@R=16pt{M \ar[r]^{\id{M}} \ar[d]_{e_u\circ u_2}& M \ar[d]^{m}\\ E_u \ar[r]_{m_u} & X}\] 
In particular, we can now apply Lemma \ref{lem:pb2} to the following
$\mathcal@R=16pt{M}, \mathcal{N}$-pushout square
\[\xymatrix@C=40pt@R=16pt{N \ar[r]^{e_u\circ u_1} \ar[d]_{e_u\circ u_1}& E_u\ar[d]^{e_1}\\
E_u \ar[r]_{e_2}& Q_{e_u\circ u_1} }\]
and to the pullback rectangles
	\[
\xymatrix@C=30pt@R=18pt{ M \ar@/^.4cm/[rrr]^{\id{M}} \ar[d]_{e_u\circ u_2}\ar[r]_{m_2} & Q_{p_1} \ar[rr]_{\upsilon_{p_1}} \ar[d]_{a\circ b} && M\ar[d]^{m} & M\ar@/^.4cm/[rrr]^{\id{M}} \ar[d]_{e_u\circ u_2} \ar[r]_{m_1}  & Q_{p_1} \ar[rr]_{\upsilon_{p_1}} \ar[d]_{a\circ b}  && M \ar[d]^{m}\\ E_u \ar@/_.4cm/[rrr]_{m_u}\ar[r]^-{e_2} & Q_{e_u\circ u_1} \ar[r]^{\upsilon_{e_u\circ u_1}} &E_u \ar[r]^{m_u} &X& E_u \ar@/_.4cm/[rrr]_{m_u}\ar[r]^-{e_1}& Q_{e_u\circ u_1} \ar[r]^{\upsilon_{e_u\circ u_1}} & E_u \ar[r]^{m_u} &X }\]
to show that the outer rectangle in the diagram below is a pullback, so that,  in particular, $a\circ b\in \mathcal{M}$. We can also 
apply Lemma \ref{lem:pb1} to deduce that the left half of the rectangle  is a pullback too. 
\[\xymatrix@C=40pt@R=16pt{Q_{p_1} \ar[r]^{\upsilon_{p_1}} \ar[d]_{a\circ b} &M \ar[r]^{\id{M}} \ar[d]_{e_u\circ u_2}& M \ar[d]^{m}\\ Q_{e_u\circ u_1} \ar[r]_-{\upsilon_{e_u\circ u_1}} & E_u \ar[r]_{m_u}& X}\]
If we compute, we can notice that
\begin{gather*}
\upsilon_{e_u\circ u_1}\circ a\circ b  = e_u\circ u_2\circ \upsilon_{p_1}= r\circ b	
\qquad
	\upsilon_{e_u\circ u_1}\circ a\circ a_1 = \upsilon_{e_u\circ u_1}\circ e_1= \id{E_u}=r\circ a_1
\end{gather*}
and therefore $r= \upsilon_{e_u\circ u_1}\circ a$.

We can apply Lemma \ref{lem:pb1} to the rectangle below, showing that its left half is a pullback
\[ \xymatrix@R=16pt{Q_{p_1} \ar[d]_b \ar[r]^{\id{Q_{p_1}}} & Q_{p_1}  \ar[d]_{a\circ b}\ar[r]^{\upsilon_{p_1}}& M \ar[d]^{e_u\circ u_2}\\ A \ar@/_.4cm/[rr]_{r}\ar[r]^-a&Q_{e_u\circ u_1} \ar[r]^{\upsilon_{e_u\circ u_1}}& E_u }\]

Suppose now that the solid part of the diagram below is given
\[\xymatrix@R=16pt{Q_{p_1} \ar[r]^{b} \ar[d]_{\upsilon_{p_1}}& A \ar[r]^-{a}& Q_{e_u\circ u_1} \ar[d]_{\upsilon_{e_u\circ u_1}} \ar@/^.3cm/[ddr]^{z_1} \\M \ar@/_.3cm/[drrr]_{z_2}\ar[r]_{u_2}&U \ar[r]_{e_u}& E_u \ar@{.>}[dr]^{z}\\ & & & Z}\]
we want to show that the inner rectangle is a pushout. Uniqueness of the dotted $z:E_u\to Z$ is guaranteed by the fact that $\upsilon_{e_u\circ u_1}$ is an epimorphism, so it is enough to construct an arrow fitting in the diagram. 

First of all we can notice that
\[z_1\circ e_1\circ e_u\circ u_1 = z_1\circ e_2\circ e_u\circ u_1\]
while, on the other hand, we have
\begin{align*}
z_1\circ e_1\circ e_u\circ u_2 &= z_1\circ a\circ a_1\circ e_u\circ u_2=z_1\circ a \circ b\circ m_1=z_2\circ \upsilon_{p_1} \circ m_1\\
& = z_2\circ \upsilon_{p_1} \circ m_2
=z_1\circ a \circ b\circ m_2 =z_1\circ a\circ a_2\circ u_2=z_1\circ e_2\circ e_u \circ u_2
\end{align*}
which implies that $z_1\circ e_1\circ e_u  = z_1\circ e_2\circ e_u$, so, since $e_u$ is an epimorphism, we have
$z_1\circ e_1=z_1\circ e_2$.
So equipped, we can now compute:
\begin{gather*}
z_1\circ e_1\circ \upsilon_{e_u\circ u_1}\circ e_1=z_1\circ e_1\circ \id{E_u}=z_1\circ e_1	\\
	z_1\circ e_1\circ \upsilon_{e_u\circ u_1}\circ e_2=z_1\circ e_1\circ \id{E_u}=z_1\circ e_1=z_1\circ e_2
\end{gather*}
showing  that $z_1$ is equal to $z_1\circ e_1 \circ  \upsilon_{e_u\circ u_1}$. 

On the other hand, 
\begin{align*}
z_2\circ \upsilon_{p_1}=z_1\circ a \circ b =z_1\circ \id{E_u} \circ a \circ b=z_1\circ e_1 \circ \upsilon_{e_u\circ u_1} \circ a\circ b=z_1\circ e_1\circ e_u\circ u_2 \circ \upsilon_{p_1} 
\end{align*}
and $\upsilon_{p_1}$ is an epimorphism and thus
$z_2=z_1\circ e_1\circ e_u \circ u_2$.
Summing up, $z_1\circ e_1$ fills our original diagram, thus its inner rectangle is indeed a pushout.

We are now ready to collect all our arrows in the following cube
		\[\xymatrix@C=20pt@R=8pt{ & &Q_{p_1}\ar[ddll]_{b} \ar[ddd]^(.4){\id{Q_{p_1}}}|(.666666)\hole\ar[rrr]^{\upsilon_{p_1}} &&&M \ar[ddd]^{\id{M}} \ar[dl]^{u_2}\\&  &&&U \ar[ddd]^{\id{U}} \ar[dl]^{e_u} \\A \ar[ddd]_{a} \ar[rrr]^{r}&&& E_u\ar[ddd]^(.6){\id{E_u}}\\&&Q_{p_1} \ar[rrr]|(.36)\hole^{\upsilon_{p_1}}|(.69)\hole \ar[dl]_{b}&&& M \ar[dl]^{u_2}\\ & A  \ar[dl]_(.4){a}&&& U \ar[dl]^{e_u} \\ Q_{e_u\circ u_1}\ar[rrr]_{\upsilon_{e_u\circ u_1}} &&& E_u}\]
This cube has an $\mathcal{M}, \mathcal{N}$-pushout as top and bottom face and all faces beside the frontal one are pullbacks, hence, by $\mathcal{M}, \mathcal{N}$-adhesivity it follows that also this last face is a pullback. By Lemma \ref{lem:pb1} the rectangles
\[\xymatrix@C=40pt@R=16pt{U\ar[r]^{a_2} \ar[d]_{e_u} & A \ar[r]^{r} \ar[d]_{a}& E_u \ar[d]^{\id{E_u}}\\
E_u \ar[r]_{e_2}& Q_{e_u\circ u_1} \ar[r]_{\upsilon_{e_u\circ u_1}} & E_u}\]
is a pullback, thus $e_u$ is an isomorphism as it is the pullback of $\id{E_u}$.
\end{proof}

\begin{cor}\label{cor:made}
Let $\X$ be a category with pullbacks and $\mathcal{M}$ a stable system of monos on it. If $\X$ is $\mathcal{M}$-adhesive, then for every object $X$  and every $[m]$ and $[n]$ in $\msub{M}(X)$, their supremum in $(\sub(X), \leq)$ exists and it belongs to $\msub{M}(X)$.
\end{cor}

Combining Theorem \ref{thm:primo} with Theorem \ref{thm:secondo} we obtain also the following results.

\begin{cor}\label{cor:equiv} Let $(\mathcal{M}, \mathcal{N})$ be a preadhesive structure on  a category $\X$ with pullbacks. If $\mathcal{M}\cap \mathcal{N}$ contains every split mono and every $\mathcal{M}, \mathcal{N}$-codiagonal is in $\mathcal{N}$, then the following are equivalent:
	\begin{enumerate}
		\item $\mathcal{M}\subseteq \ad{N}$ and $\mathcal{M}$ is closed under $\mathcal{M}$, $\mathcal{N}$-unions;
		\item $\X$ is $\mathcal{M}, \mathcal{N}$-adhesive.
	\end{enumerate}
\end{cor}

Finally Proposition \ref{prop:po} and Corollaries \ref{cor:mad} and \ref{cor:mads} yield the result below.

\begin{cor}\label{cor:mequiv}
	Let $\mathcal{M}$ be a stable system of monos on a category $\X$ with pullbacks and suppose that $\mathcal{M}$ contains all split monos, then the following are equivalent:
\begin{enumerate}
	\item $\X$ is $\mathcal{M}$-adhesive;
	\item every $\mathcal{M}$-pushout square is Van Kampen and for every object $X$, any two $[m], [n]\in \msub{M}(X)$ have a supremum in $(\sub(X), \leq)$ belonging to $\msub{M}(X)$;
		\item $\mathcal{M}$ is stable under pushouts, $\mathcal{M}\subseteq \mor(X)_{\mathsf{a}}$ and for every object $X$, every $[m], [n]\in \msub{M}(X)$ have a supremum in $(\sub(X), \leq)$ which is again in $\msub{M}(X)$.
\end{enumerate}
\end{cor}

\begin{rem}\label{rem:ex}
	Notice that, in items $2$ and $3$ of the previous corollary, the existence of a supremum in $(\sub(X), \leq)$ for $[m], [n]\in \msub{M}(X)$ is guaranteed by the hypothesis that every arrow in $\mathcal{M}$ is adhesive and by Proposition \ref{prop:uni}.
\end{rem}

\section{An embedding theorem for $\mathcal{M}, \mathcal{N}$-adhesive categories} \label{sec:top}
In this section we generalize the embedding theorem of \cite{garner2012axioms,lack2011embedding}: given a category $\X$ which is $\mathcal{M}, \mathcal{N}$-adhesive, if $\mathcal{M}$ and $\mathcal{N}$ are nice enough, we are able to construct a (Grothendieck) topos in which $\X$ embeds via a functor preserving pullbacks and $\mathcal{M}, \mathcal{N}$-pushouts.

\begin{defi} 
Let $(\mathcal{M}, \mathcal{N})$ be a preadhesive structure for a category $\X$. A \emph{$\jp$-covering family}  for an object $X$ is a set $\{p, q\}$ of arrows $p:Z\to X$ and $q:Y\to X$ such that there exist $m:N\to Y$ in $\mathcal{M}$ and $n:N\to Z$ in $\mathcal{N}$ making the following square a pushout
\[\xymatrix@R=16pt{N \ar[r]^{n} \ar[d]_{m}& Z \ar[d]^{p} \\ Y \ar[r]^q & X}\]
We will define $\jp(X)$ as the set of $\jp$-covering families for $X$.
\end{defi}

\begin{prop}Let $\X$ be a category with all pullbacks. Given a preadhesive structure  $(\mathcal{M}, \mathcal{N})$ such that $\mathcal{M}\subseteq \ad{N}$, the family $\{\jp(X)\}_{X\in \X}$ defines a \emph{coverage} (in the sense of \cite[Def.~C2.1.1]{johnstone2002sketches}) $\jp$ on $\X$.
\end{prop}

\begin{rem}The coverage $\jp$ is a \emph{cd-structure} in the sense of \cite{VOEVODSKY20101384,voevodsky2010unstable}.
\end{rem}

Let us recall some terminology from basic sheaf theory (see, e.g., \cite{johnstone2002sketches,maclane2012sheaves}). Let $(\X, \j)$ be a site and $\{f_i\}_{i\in I}$ be a covering family for $X$, i.e., $f_i:Y_i\to X$ and $\{f_i\}_{i\in I}$ belongs to $\j(X)$. Given a functor $\X^{op}\to \Set$,  a family $\{a_i\}_{i\in I}$ of elements $a_i\in F(Y_i)$ is \emph{compatible} for $F$ if whenever we have two arrows $g:U\to X_i$, $h:U\to X_j$ such that $f_i\circ g = f_j\circ h$ 
then $F(g)(a_i)=F(h)(a_j)$.
A functor $F:\X^{op}\to \Set$ is a \emph{sheaf} if for every covering family $\{f_i\}_{i\in I}$ and compatible family $\{a_i\}_{i\in I}$ there exists a unique \emph{amalgamation} $a\in F(X)$ such that, for every $i\in I$,  $F(f_i)(a)=a_i$.

Our next step is to characterize sheaves for the site $(\X, \jp)$.

\begin{lem}\label{lem:cond}Let $(\mathcal{M}, \mathcal{N})$ be a preadhesive structure for a category $\X$ with pullbacks and suppose that every element in $\mathcal{M}$ is $\mathcal{N}$-adhesive. Then the following are equivalent for a presheaf $F:\X^{op}\to \Set$:
	\begin{enumerate}
	\item $F$ is in $\sh(\X, \jp)$;
	\item given the following two squares, the first of which is a  $\mathcal{M}, \mathcal{N}$-pushout and the second one a pullback, 
	\[\xymatrix@R=16pt{N \ar[r]^{n} \ar[d]_{m}& Z \ar[d]^{p} & K_q \ar[r]^{y_1}  \ar[d]_{y_2}& Y\ar[d]^{q} \\ Y \ar[r]_q & X & Y\ar[r]_q & X }\]
	if the solid part of the diagram below is given, then there exists a unique $f:S\to F(X)$ fitting in it.	
	\[\xymatrix@C=40pt{S\ar@{.>}[dr]^{f} \ar@/_.3cm/[ddr]_{f_1} \ar@/^.3cm/[drr]^{f_2}\\& F(X) \ar[d]_{F(p)} \ar[r]^{F(q)} & F(Z) \ar[d]^{F(m)}\\ F(K_q) & F(Y) \ar[r]_{F(n)} \ar@<.5ex>[l]^(.455){F(y_2)}  \ar@<-.5ex>[l]_{F(y_1)} & F(N)}\]
	\end{enumerate} 
\end{lem}
\begin{proof}$(1\Rightarrow 2)$ Let us start showing that, for every $s\in S$, the family $\left\{f_1(s), f_2(s)\right\}$ is compatible for $F$. Given a $\jp$-covering family $\{p,q\}$ and three commutative squares as the ones below
	\[\xymatrix@R=16pt{A \ar[r]^{a_1} \ar[d]_{a_2}& Y \ar[d]^{q} & B \ar[r]^{b_1} \ar[d]_{b_2}& Y \ar[d]^{q} & C \ar[r]^{c_1} \ar[d]_{c_2}& Z \ar[d]^{p}\\ Z \ar[r]_{p} & X & Y \ar[r]_{q} & X & Z \ar[r]_{p} & X}\]
$p$ is the pushout of $m$, thus it belongs to $\mathcal{M}$ and so is a mono, implying that
\[c_1=c_2 \qquad F(c_1)\circ f_2=F(c_2)\circ f_2\]
Moreover $m\in \ad{N}$ and $n\in \mathcal{N}$, thus in the following diagrams the two inner squares are pullbacks,  giving us the dotted arrows $a:A\to N$, $b:B\to K_q$.
\[\xymatrix@R=16pt{A \ar@/^.4cm/[drr]^{a_2} \ar@/_.4cm/[ddr]_{a_1} \ar@{.>}[dr]^{a}&&& B \ar@/^.4cm/[drr]^{b_2} \ar@/_.4cm/[ddr]_{b_1} \ar@{.>}[dr]^{b}\\&N \ar[r]^{n} \ar[d]_{m}& Z \ar[d]^{p} && K_q \ar[r]^{y_2} \ar[d]_{y_1}& Y \ar[d]^{q}\\ &Y \ar[r]_{q} & X && Y \ar[r]_{q} & X}\]
Computing we get the following chains of identities
\begin{gather*}
	F(a_1)\circ f_1= F(a)\circ F(m)\circ f_1= F(a)\circ F(n)\circ f_2=F(a_2)\circ f_2\\
	F(b_1)\circ f_1= F(b)\circ F(y_1)\circ f_1=F(b)\circ F(y_1)\circ f_1=F(b_2) \circ f_1 
\end{gather*}
which imply that, for every $s\in S$,  $\{f_1(s), f_2(s)\}$ is a compatible family for $F$. Since $F$ is a sheaf we can define $f:S\to F(X)$ taking as $f(s)$ the unique amalgamation of $\{f_1(s), f_2(s)\}$, by construction
\[f_1=F(p)\circ f \qquad f_2=F(q)\circ f \]
For uniqueness, it is enough to notice that, if $g:S\to F(X)$ is another arrow such that 
\[f_1=F(p)\circ g \qquad f_2=F(q)\circ g\]
 then $g(s)$ is an amalgamation for $\{f_1(s), f_2(s)\}$.
 
	\smallskip\noindent $(2\Rightarrow 1)$ Let $\{p,q\}$ be a $\jp$-cover of $X$, by definition there exists an $\mathcal{M}, \mathcal{N}$-pushout square 
		\[\xymatrix@R=16pt{N \ar[r]^{n} \ar[d]_{m}& Z \ar[d]^{p} \\ Y \ar[r]_q & X}\]
		Take  a compatible family $\{s_1, s_2\}$, with $s_1\in F(Y)$ and $s_2\in F(Z)$.  From  the two squares
		\[\xymatrix@R=16pt{N \ar[r]^{n} \ar[d]_{m}& Z \ar[d]^{p} & K_q \ar[r]^{y_1}  \ar[d]_{y_2}& Y\ar[d]^{q} \\ Y \ar[r]_q & X & Y\ar[r]_q & X }\]
		we get that $F(m)(s_1)=F(n)(s_2)$ and $F(y_1)(s_1)=F(y_2)(s_1)$.
		Thus, if $\delta_{s_1}:1\to F(Y)$ and $\delta_{s_2}:1\to F(Z)$ pick, respectively, $s_1$ and $s_2$, we have the solid part of the following commutative diagram and, by hypothesis, also the dotted $\delta:1\to F(X)$.		
		\[\xymatrix@C=40pt{1\ar@{.>}[dr]^{\delta} \ar@/_.3cm/[ddr]_{\delta_{s_1}} \ar@/^.3cm/[drr]^{\delta_{s_2}}\\& F(X) \ar[d]_{F(p)} \ar[r]^{F(q)} & F(Z) \ar[d]^{F(m)}\\ F(K_q) & F(Y) \ar[r]_{F(n)} \ar@<.5ex>[l]^(.455){F(y_2)}  \ar@<-.5ex>[l]_{F(y_1)} & F(N)}\]
		Now, let $s$ be the element of $F(X)$ picked by $\delta$, then, by construction $s$ is an amalgamation for $\{s_1, s_2\}$. On the other hand, if $s'$ is another amalgamation, then $\delta_{s_1}=F(p)$ and $\delta_{s'} \qquad \delta_{s_2}=F(p)\circ \delta_{s'}$, so that $\delta=\delta_{s'}$, showing that $s=s'$, i.e. that $F$ is a sheaf.
\end{proof}

We can now combine the previous lemma with Lemma \ref{lem:fon} to get the following.
\begin{lem}\label{lem:cond2}
Let $(\mathcal{M}, \mathcal{N})$ be a preadhesive structure on a category $\X$ with pullbacks such that  $\mathcal{M}\subseteq \ad{N}$, $\mathcal{M}\cap\mathcal{N}$ contains every split mono  and  $\mathcal{M}$ is closed under $\mathcal{M},\mathcal{N}$-unions. Then for a presheaf $F:\X^{op}\to \Set$ the following are equivalent
\begin{enumerate}
	\item $F$ is in $\sh(\X, \jp)$;
	\item $F$ sends $\mathcal{M}, \mathcal{N}$-pushouts to pullbacks.
\end{enumerate}
\end{lem}
\begin{proof}
	\smallskip\noindent $(1\Rightarrow 2)$ Given the following two squares, the first one of which is an $\mathcal{M}, \mathcal{N}$-pushout, while the second is a pullback
	\[\xymatrix@R=16pt{N \ar[r]^{n} \ar[d]_{m}& Z \ar[d]^{p} & K_n \ar[r]^{n_1} \ar[d]_{n_2}& N\ar[d]^{n} \\ Y \ar[r]_q & X & N \ar[r]_n & Y }\]
	 Lemma \ref{lem:fon} gives the following diagrams, in wich the common square on the left is an $\mathcal{M}, \mathcal{N}$-pushout. In particular this implies $\{k, \gamma_{q}\}$ is a $\jp$-covering family of $X$
	\[\xymatrix@R=16pt{N \ar[r]^-{\gamma_n} \ar[d]_{m} & K_n \ar[d]_{k} \ar[r]^{n_1}& N \ar[r]^n \ar[d]_{m}& Z\ar[d]^{p} & N \ar[r]^-{\gamma_n} \ar[d]_{m}& K_n \ar[d]_{k} \ar[r]^{n_2}& N \ar[r]^n \ar[d]_{m}& Z\ar[d]^{p}\\Y \ar[r]_-{\gamma_q}& K_q \ar[r]_{y_1}& Y\ar[r]_{q} & X & Y \ar[r]_-{\gamma_q} & K_q \ar[r]_{y_2}& Y\ar[r]_{q} & X}\]
	$k$ is in $\mathcal{M}$ since it is the pushout of $\mathcal{M}$, thus $k$ and $\gamma_q$ are both mono, so that Lemma \ref{lem:cond} now implies that the square below is a pullback.
	\[\xymatrix@R=16pt{F(K_q) \ar[r]^{F(k)} \ar[d]_{F(\gamma_q)} & F(K_n) \ar[d]^{F(\gamma_n)}\\ F(Y) \ar[r]_{F(m)} & F(N)}\]

 Suppose that the solid part of the following diagram is given
	\[\xymatrix@C=30pt@R=16pt{S\ar@{.>}[dr]^{f} \ar@/_.5cm/[ddr]_{f_1} \ar@/^.5cm/[drr]^{f_2}\\& F(X) \ar[d]_{F(p)} \ar[r]^{F(q)} & F(Z) \ar[d]^{F(n)}\\ & F(Y) \ar[r]_{F(m)}  & F(N)}\]

On one hand we have at once that
\begin{align*}	F(\gamma_q)\circ F(y_1)\circ f_1=F(y_1\circ \gamma_q)\circ f_1=F(\id{Y})\circ f_1=F( y_2\circ \gamma_q)\circ f_1=	F(\gamma_q)\circ F(y_2)\circ f_1
\end{align*}
while, on the other hand, we also have
 \begin{align*}
	F(k)\circ F(y_1)\circ f_1 &=F( y_1 \circ k)=F(m\circ n_1)\circ f_1 =F(n_1)\circ F(m)\circ f_1\\
	&=F(n_1)\circ F(n)\circ f_2=F(n_2)\circ F(n)\circ f_2=F(n_2)\circ F(m)\circ f_1\\
	&=F(m\circ n_2) \circ f_1=F(y_2\circ k) \circ f_1=F(k)\circ F(y_2)\circ f_1
\end{align*}
and we  can deduce that $F(y_1)\circ f_1= F(y_2)\circ f_1$, thus Lemma \ref{lem:cond} now entails the thesis.
	
	\smallskip\noindent $(2\Rightarrow 1)$ This follows at once from Lemma \ref{lem:cond}.
\end{proof}

We are now ready to obtain our main theorem.

\begin{thm} 
	Let $(\mathcal{M}, \mathcal{N})$ be a preadhesive structure on a category $\X$ with pullbacks such that  $\mathcal{M}\subseteq \ad{N}$, $\mathcal{M}\cap\mathcal{N}$ contains every split mono  and  $\mathcal{M}$ is closed under $\mathcal{M},\mathcal{N}$-unions. Then
	\begin{enumerate}
		\item the Yoneda embedding $\yo_{\X}:\X\to \Set^{\X^{op}}$ factors through a full and faithful functor $\yo'_{\X}:\X \to \sh(\X, \jp)$;
		\item $\yo'_{\X}$ preserves pullbacks and sends $\mathcal{M}, \mathcal{N}$-pushouts to pushouts.
	\end{enumerate}
\end{thm}
\begin{proof}
	\begin{enumerate}[leftmargin=0pt,itemindent=1.7em]
	\item  Since $\sh(\X, \jp)$ is a full subcategory of $\Set^{\X^{op}}$, it is enough to show that, for every $X\in \X$, the functor $\X(-, X)$ is a sheaf, but this follows at once from Lemma \ref{lem:cond2}, since any representable presheaf sends pushouts to pullbacks.
	\item The inclusion $\sh\left(\X, \jp\right)\to \Set^{\X^{op}}$ creates limits and $\yo_\X$ sends pullbacks in pullbacks, therefore $\yo'_\X$ preserves pullbacks too. Take now an $\mathcal{M}, \mathcal{N}$-pushout 
		\[\xymatrix@R=16pt{X \ar[r]^n \ar[d]_{m}& Z \ar[d]^{q}\\ Y\ar[r]_p &Q}\]
Since $\sh\left(\X, \jp\right)$ is a full subcategory of $\Set^{\X^{op}}$, by Yoneda lemma every sheaf $F$ yields a natural isomorphism $\mathsf{y}:\sh\left(\X, \jp\right)\left(\yo'_{\X}(-), F\right)\to F$, so we obtain a diagram
\[\xymatrix@C=30pt@R=16pt{\mathbf{Sh}\left(\X, \jp\right)\left(\yo'_{\X}(Q), F\right) \ar[dr]_{\mathsf{y}_Q} \ar[rrr]^{(-)\circ \yo'_{\X}(q)} \ar[ddd]_{(-)\circ \yo'_{\X}(p)}& & & \sh\left(\X, \jp\right)\left(\yo'_{\X}(Z), F\right) \ar[dl]^{\mathsf{y}_Z} \ar[ddd]^{(-)\circ \yo'_{\X}(n)}\\& F(Q)  \ar[r]^{F(q)} \ar[d]_{F(p)}& F(Z) \ar[d]^{F(n)} \\& F(Y) \ar[r]_{F(m)} & F(X) \\ \sh\left(\X, \jp\right)\left(\yo'_{\X}(Y), F\right)  \ar[ur]^{\mathsf{y}_Y}\ar[rrr]_{(-)\circ \yo'_{\X}(m)}& & & \sh\left(\X, \jp\right)\left(\yo'_{\X}(X), F\right) \ar[ul]_{\mathsf{y}_X}}\]
$F$ is a sheaf, hence the inner square is a pullback by Lemma \ref{lem:cond2} and thus the outer one is a pullback too, proving that $\yo'_\X$ sends $\mathcal{M}, \mathcal{N}$-pushouts to pushouts.
	\qedhere 
	\end{enumerate}
\end{proof}

\begin{cor}
Let $\X$ be an $\mathcal{M}, \mathcal{N}$-adhesive  category with pullbacks such that $\mathcal{M}\cap \mathcal{N}$ contains all split monomorphisms and $\mathcal{N}$ contains all $\mathcal{M}, \mathcal{N}$-codiagonals, then there exists a full and faithful functor from $\X$ into a topos, preserving all pullbacks and  $\mathcal{M}, \mathcal{N}$-pushouts.
\end{cor}

\section{Conclusions and further work}\label{sec:concl}
In this paper we have generalized three well-known results from the theory of (quasi)adhesive categories to the $\mathcal{M}, \mathcal{N}$-adhesive setting, adapting the techniques developed in \cite{garner2012axioms}. 

First, we have studied the poset of subobjects of an $\mathcal{M}, \mathcal{N}$-adhesive category. We have shown that given a mono in $\mathcal{M}$ and one in $\mathcal{M}\cap \mathcal{N}$, then their supremum, called a \emph{$\mathcal{M}, \mathcal{N}$-union}, exists and it is computed as the pushout of the pullbacks of the given monos.

Next, we have proved a kind of converse of the previous result: in  presence of $\mathcal{M}, \mathcal{N}$-unions, we can guarantee $\mathcal{M}, \mathcal{N}$-adhesivity once we know that $\mathcal{M}$ is contained in the class of $\mathcal{N}$-adhesive morphisms. This allows the reduction of the proof of the Van Kampen condition to the proof of stability of some square and of some pullbacks being pushouts. As an example, adhesivity of toposes can be readily proved using this methods.

Finally, we have shown that, under some mild hypotheses about $\mathcal{M}$ and $\mathcal{N}$, an $\mathcal{M}, \mathcal{N}$-adhesive category can be embedded in a Grothendieck topos via a functor preserving all the relevant structure (i.e. pullbacks and $\mathcal{M}, \mathcal{N}$-pushouts). Therefore, the slogan for which ``an adhesive category is one whose pushouts of monomorphisms exist and behave more or less as they do in a topos'' holds even for $\mathcal{M}, \mathcal{N}$-adhesive categories.

An open question is how to equip a category $\X$ with a preadhesive structure $(\mathcal{M}, \mathcal{N})$, suited to a given situation.
One may notice that if a category is $\mathcal{M}, \mathcal{N}$-adhesive, then $\mathcal{M}$ must be contained in the class of $\mathcal{N}$-adhesive morphisms. In particular, in the $\mathcal{M}$-adhesive case, $\mathcal{M}$ must be a subclass of the class of adhesive morphisms.
This suggests to investigate the poset of all preadhesive structures $(\mathcal{M}, \mathcal{N})$, with $\mathcal{M}$ contained in $\ad{N}$, in order to characterize the largest one, suited for the specific problem, for which $\X$ is $\mathcal{M}, \mathcal{N}$-adhesive.%

\bibliographystyle{alphaurl}
\bibliography{bibliog.bib}
\end{document}